%% file: main.tex
\documentclass[12pt,a4paper]{article}
\include{preamble}

\title{\bfseries 
Can Reinforcement Learning Efficiently Discover Price Manipulation?}

\author[]{Ioanna-Yvonni Tsaknaki}
\author[]{Andrea Macrì}
\author[]{Fabrizio~Lillo}

\affil[]{\footnotesize Scuola Normale Superiore, Piazza dei Cavalieri 7, Pisa 56126, Italy.}
\date{}

\begin{document}
\maketitle
\begin{abstract}

In this paper, we investigate whether a model-free RL agent can identify and exploit price manipulation opportunities more effectively than a traditional model-based approach that assumes correct specification of the data-generating process but relies on noisy parameter estimates. We consider a single-asset market in which prices evolve according to an Almgren–Chriss framework with non-linear permanent impact and linear temporary impact. We first establish the existence of price-manipulative strategies in discrete time and compute the optimal benchmark strategy using Sequential Least Squares Quadratic Programming under full information. We then compare two finite-sample learning approaches: a model-based procedure that estimates impact parameters from simulated execution data and an agnostic RL approach based on Deep Deterministic Policy Gradient, trained directly on the same amount of data. For intermediate volatility, the RL agent successfully discovers profitable manipulative strategies without explicit knowledge of the underlying model, even when training data are quite limited. More importantly, RL consistently outperforms the model-based approach when parameter estimates are affected by sampling error, despite the latter benefiting from the correct model specification. For large volatility, all methods are unable to identify manipulation opportunities, while for small volatility, the model based approach outperforms RL. These findings highlight both the effectiveness of RL in complex control problems and the risks associated with deploying learning algorithms in financial markets without appropriate safeguards.

\end{abstract}
\section{Introduction}
\label{sec:intro}
 The increasing use of autonomous algorithmic trading systems has transformed modern financial markets. Reinforcement Learning (RL) methods, in particular, have attracted growing attention because of their ability to learn adaptive trading policies directly from market interaction. While these techniques offer significant advantages in high-dimensional and partially observed environments, they also raise important concerns for regulators. Since RL agents optimize rewards without explicit economic understanding or regulatory awareness, they may discover trading behaviours that are profitable but potentially manipulative. As emphasized by \cite{wellman2017ethical} and \cite{Azzutti_Law} autonomous trading agents may create new forms of market abuse and challenge traditional regulatory frameworks designed around human intent and accountability. Recent research in market manipulation and market microstructure has highlighted the vulnerability of electronic markets to algorithmic strategies capable of distorting price formation and exploiting structural weaknesses of limit order book environments \cite{lee2013microstructure}. In modern electronic markets, even relatively simple algorithmic trading rules can artificially influence short-term prices and induce predictable reactions from other market participants. Such behaviours may undermine market integrity and generate misleading signals regarding supply and demand. Moreover, recent studies suggest that reinforcement learning agents may autonomously discover manipulative trading strategies purely through reward maximization and repeated interaction with the environment \cite{martinez2016learning,mizuta2020ai}. Specifically in \cite{martinez2016learning} they predefine honest and manipulative actions, such as spoofing and pinging and use RL to determine when those already-defined actions are optimal under different transaction costs, penalties, and liquidity conditions. The actual price-manipulation mechanism is assumed in the model rather than learned by the agent. On the other hand, in \cite{mizuta2020ai} a genetic algorithm discovers a pump-and-dump pattern through feedback with trend-following agents in an artificial market. This differentiates us since we explicitly model non-linear permanent market impact and show that a model-free DDPG agent can recover such strategies in a continuous-action environment.

 One potential source of manipulation arises from the presence of non-linear market impact, which is the reaction of prices to trades. A large empirical literature such as \cite{Lillo_2003_nature,Potters_2003,10.1007/978-3-642-17045-4_8} and \cite{Zhou01082012} documents that instantaneous market impact is typically a non-linear function of volume for individual trades, typically with an odd and concave shape. While these effects are well established empirically, they may also create environments in which trading activity can influence prices in exploitable ways. In particular, non-linear impact can reintroduce the possibility of dynamic arbitrage\footnote{Dynamic arbitrage exists if a trader can buy and sell over time, end with no inventory, and still make positive expected profit purely from the market-impact dynamics.}, allowing round-trip trading strategies, those with zero initial and terminal inventory, to generate positive expected cash flow purely through their impact on prices. This phenomenon, commonly referred to as price manipulation, was formalized by \cite{Huberman_2004} and further explored in continuous-time settings by \cite{Gatheral_2010} and \cite{Gueant_2016}. Understanding when and how such manipulation arises remains a central question in market microstructure.

Although  one can often find mathematically price manipulations for a given impact model, their identification from market data poses several challenges. First, the true Data Generating Process (DGP) of the price and trade dynamics is unknown, and models represent at most a stylized description of reality. Second, even when the DGP is known, its parameters must be estimated from a finite sample of data. The error in their estimation can severely undermine the manipulation opportunities. Thus, the mere existence of manipulations does not imply a systematically profitable strategy.

RL offers an intriguing alternative since it is learning based and model-free.
Rather than postulating a model and estimating the parameters on a finite sample, it uses this sample to directly identify the policy that maximizes the reward. It is therefore natural to ask whether RL based manipulation identification outperforms in finite sample the  model based approach, even under the correct DGP specification. This is the main question we answer in this paper.

The use of RL methods for optimal execution and trading, where actions correspond to trading rates, is widespread, both in academia and in industry. The Deterministic Policy Gradient framework introduced by \cite{pmlr-v32-silver14} and its deep learning extension, Deep Deterministic Policy Gradient (DDPG) by \cite{pmlr-v14-lillicrap}, enable efficient learning in continuous action spaces. These methods have been successfully applied to optimal execution and portfolio management tasks (see, e.g., \cite{Jiang2017,Moody2001, micheli2024deep}). 
Early applications of RL to trading include \cite{Nevmyvaka2006}, who applied Q-learning to execution problems, and more recent work such as \cite{Hendricks2014} and \cite{Ning_2021}, which demonstrate the effectiveness of deep RL methods in complex market environments. Extensions incorporating realistic features such as stochastic liquidity and partial observability have been studied in \cite{Macri_2025} and \cite{Macri_2026}, respectively.

The use of RL in financial markets presents several criticalities, since algorithmic trading systems must operate within strict regulatory frameworks designed to prevent manipulation and ensure market integrity. In fact, as pointed out in a recent work by  \cite{Cartea_Chang_GarciaArenas}, the risk of unintentional market manipulation exists, especially when learning algorithms converge to strategies that artificially influence prices in order to increase profits. This phenomenon is closely related to the broader problem of reward misspecification in machine learning, where the objective function does not fully capture the desired behaviour of the agent. In the context of optimal execution, an RL agent trained to maximize trading profits may inadvertently learn to exploit non-linear impact effects, effectively engaging in dynamic arbitrage even when such behaviour is undesirable or unrealistic, simply following a reward maximisation logic while learning from the environment.

Motivated by these considerations, this paper examines whether an RL agent is more effective than a traditional model-based strategy, which assumes correct model specification but faces parameter uncertainty arising from finite-sample estimation, in generating positive expected cash flows through price manipulation in a market characterized by non-linear price impact.
 We consider a single-asset setting, where the price evolves according to an Almgren-Chriss (AC) type model \citep{AC} with non-linear permanent impact and linear temporary impact (i.e. quadratic temporary impact costs). The agent is constrained to follow a round-trip strategy, meaning that it starts and ends with zero inventory over a finite trading horizon. This constraint ensures that any profit must arise purely from trading dynamics rather than directional exposure. Leveraging on this framework, and extending \cite{Gueant_2016}, by considering discrete time and non-vanishing temporary impact costs,
we first prove the existence of manipulative strategies from a theoretical perspective, showing that non-linear impact functions can indeed admit solutions with positive expected cash flow even when only two trading velocities are allowed. Due to the non-convexity of the optimization problem, we numerically find the general optimal strategy by using Sequential Least Squares Quadratic Programming (SLSQP), assuming full knowledge of the model parameters. The SLSQP algorithm has been used previously for the solution of optimal execution problem for the non linear transient impact model in \cite{Curato_2017}. This provides a benchmark solution against which learning-based methods can be compared. 

We then consider the problem of learning to manipulate the price from a finite sample of data and we compare two approaches. The first is the model based one: we use a finite sample of simulated TWAP execution to estimate the temporary and permanent impact parameters following the approach proposed by \cite{Almgren_Thum_Hauptmann_Li}. Notice that in this approach the correct DGP is known and only the parameters are estimated. We then use them to implement the optimal manipulation strategies obtained with the SLSQP optimization on a test set of simulations.  The second approach is model free and based on RL: we train a DDPG algorithm to learn directly to manipulate the price on a finite sample of the same size of the one of TWAP executions. Then, as before, we test the learned strategy on a test set. For both approaches, the critical variable is volatility 
since, depending on its level, it prevents
an accurate estimation of impact parameters or an effective learning of the strategy, respectively.

Our first finding is that  the RL agent is capable of discovering trading strategies that yield positive expected cash flow, effectively identifying price-manipulative behaviors without explicit knowledge of the model, also when the training set is relatively small and volatility is not too large.
More importantly, the numerical results show that the performance of the RL-based strategy is superior to the one of the model based approach when impact parameters are noisily estimated from data. The superiority of RL is observed when volatility is not too low, because otherwise the parameter estimation of the model based approach is extremely accurate and therefore the method is equivalent to an exact optimization of a well specified DGP.  It is important to stress that the central question of the paper is therefore not whether manipulation exists theoretically, but whether an RL algorithm can independently learn to exploit these market vulnerabilities, when they are present, and if it can outperform a traditional model based approach.

These results highlight both the effectiveness of RL in solving complex control problems and the potential risks associated with its unsupervised deployment in financial markets. In particular, they underscore the need for a careful design of reward functions and regulatory safeguards to prevent potentially \textit{unintended} market manipulation adopted by learning algorithms. From a broader perspective, this work extends the literature on non-linear market impact and dynamic arbitrage by demonstrating how modern learning algorithms can exploit these features in practice. At the same time, it connects to the growing body of research on the safety and robustness of machine learning systems, emphasizing the importance of aligning algorithmic objectives with market stability and regulatory requirements.

The rest of the paper is organized as follows. In Section \ref{sec: Frame the problem}, we formulate the problem and demonstrate the existence of a manipulative solution involving only two trading velocities. In Section \ref{subsec: SLSQP}, we introduce the SLSQP numerical optimizer and present its manipulative solution under full information of the trading environment. We also illustrate the complexity of the problem and assess its performance under partial information. In Section \ref{sec:DDPG}, we introduce the DDPG algorithm. Finally, in Section \ref{sec:comparison}, we compare all the methods.
\section{Framework of the Problem}\label{sec: Frame the problem}
In this section, we provide a general framework for the problem of dynamic arbitrage. We define the market model with non-linear permanent impact in discrete time, the trading cost structure, and finally we provide a solution to the case where the trader can trade with just two velocities. 

We consider the problem of a risk-neutral agent seeking an optimal trading strategy over a finite time horizon, subject to zero initial and terminal inventory. The agent is allowed to both buy and sell during the trading period, and therefore aims to construct a \textit{round-trip} strategy. The trading horizon is divided into $T$ subintervals, each assumed, for simplicity, to have length $\Delta t=1$.

We assume that the price dynamics follow an AC model with non-linear permanent impact, namely:
\begin{align}\label{eq:fund price}
    S_t  & = S_{t-1}+f(v_t)\Delta t+\sigma\sqrt{\Delta t} \xi_t,~~~\xi_t\sim\mathcal{N}(0,1)\\\label{eq: exec price}
    \Tilde{S}_{t} & = S_{t-1}+\kappa v_{t},~~~~~~\kappa>0
\end{align}
where $S_t$ is the market price and $\Tilde{S}_t$ is the execution price at time $t$.

Assuming that the trader trades at each $t \in [T]$,\footnote{We use the notation $[T]:=\{1,\cdots,T\}$.} we denote by $q_t$ and $v_t$ the inventory and the trading velocity at time $t$ respectively.  The trading velocity $v_t$ represents the traded volume at time $t$. Clearly, it holds that $q_0=q_T=0$. The permanent market impact is non-linear and it is
\begin{equation}\label{eq: market impact}
    f(v_t) = \theta~\text{sign}(v_t)|v_t|^{\delta},~~~~~~\theta>0,\delta\in(0,1)
\end{equation}
when $\delta=1$, the permanent impact is linear as in the standard AC framework. The coefficient $\kappa$ is the temporary impact coefficient, hence we assume a linear temporary impact. We observe that a buy order pushes the price up, both the market and the execution price, since $f(v_t)>0$ and $\kappa v_t>0$. Similarly, a sell order pushes both prices down.

Now, as in other optimal execution problems, we have that:
\begin{equation*}
    \sum_{t=1}^Tv_t = q_0 = 0.
\end{equation*}
In such a problem, the final cash flow, will be:
\begin{equation*}
    X_T = -\sum_{t=1}^Tv_t\Tilde{S}_t
\end{equation*} 
 The expected cash flow is then (see Appendix \ref{app: Frame the problem}):
\begin{equation}\label{eq: final cash-flow_main text}
    \E[X_T] = -\Big[\sum_{t=1}^{T-1} f(v_t)\sum_{j=t+1}^Tv_j+\kappa\sum_{t=1}^T v_t^2\Big].
\end{equation}

As the main aim of the paper is to study the existence, feasibility and learnability of dynamic arbitrages, we briefly recall its main definition used throughout the text.
\begin{definition}[\cite{Huberman_2004}]
    In $T\in\N$ trading steps there is a \textit{dynamic arbitrage} if there exists $\{v_t\}_{t\in[T]}$ such that
    \begin{itemize}
        \item $\sum_{t=1}^Tv_t=0$ and
        \item $\E[X_T]>0$.
    \end{itemize}
\end{definition}

Therefore, finding an optimal solution to the trading problem defined above amounts to finding a solution 
$\vec v_*\in\R^T$ such that the function $\E[X_T]:=c(\vec{v}_*)$ is maximized, where by $c(\vec v_*)$ we denote a function of the optimal trading trajectory $\vec v_*$. To this end, we first need to make sure that there exists a solution to the problem we pose. The following proposition guarantees that the problem is not ill-posed.
\begin{proposition}\label{prop1}
    The expected cash flow for $\delta< 1$, is a continuous function of the strategy. Moreover searching for a strategy that maximizes the expected cash flow is a well-posed problem since
    \begin{equation*}      \lim_{\lambda\rightarrow\pm\infty}c(\lambda \vec v) = -\infty.
    \end{equation*}
    Therefore, there exists a solution that maximizes $\E[X_T]$ with $\delta< 1$. 
\end{proposition}
\begin{proof}
    See Appendix \ref{app: Frame the problem}.
\end{proof}
The above proposition does not guarantee that there is a round-trip strategy but that the problem we want to solve is not ill-posed, since the best strategy might still be to not trade.


\section{Dynamic arbitrage strategies}\label{subsec: SLSQP}

In this section, we investigate the existence and characteristics of dynamic arbitrage strategies of the previously introduced impact model. Due to the non-convexity of the problem, the investigation is mainly numerical. However in the next subsection we show that, by restricting to trading strategies with only two velocities, it is possible to find and characterize price manipulation strategies.

\subsection{Price-manipulation  with only two trading velocities}
\label{subsec:twoVelo}
Before considering more complex cases, we start by analysing the case where the trader can only select two trading velocities. Leveraging on \cite{Gueant_2016}, we denote by $\mathcal{A} \subset [-b,b]\subset\R$ the set of all possible trading actions across time step and by $|\mathcal{A}|$ the cardinality of $\mathcal{A}$. Notice that for a trading strategy $\{v_t\}_{t\in[T]}$ there are at most $T$ different trading velocities thus $|\mathcal{A}|\leq T$. We observe that for a round-trip strategy it must have $T\geq 2$ (i.e. at least two time steps). In the rest of this subsection we consider that $|\mathcal{A}|=2$.
In this case it is
\begin{equation*}
    v_t = \begin{cases}
    \alpha,~~\text{if}~t\in S\subset[T]\\
        -\beta,~~\text{if}~t\in [T]\setminus S
    \end{cases}
\end{equation*}
where $S$ is any set of indexes in $[T]$. To be a round-trip strategy, $\alpha$ and $\beta$ must have the same sign and it should be that:
\begin{equation}
\label{eq:constraints_pump_dump}
    \sum_{t=1}^T v_t = |S|\alpha-(T-|S|)\beta=0\Leftrightarrow |S|\alpha = (T-|S|)\beta.
\end{equation} 
In the following lemma, we compute the expected cash flow when one trading velocity is used up to a certain time and then replaced by a different trading velocity for the remainder of the trading horizon.   
\begin{lemma}\label{lem: expected cash}
    Let $|\mathcal{A}|=2$ and $n\in[T-1]$
         such that $v_t=\alpha$ for every $1\leq t\leq n$ and $v_t=-\beta$ for every $n< t\leq T $,  then
    \begin{equation}\label{eq:pump_dump_cash}
       \E[X_T] = \frac{n\alpha}{2}[(n+1)f(\alpha)+f(-\frac{n}{T-n}\alpha)(T-n-1)]-\kappa[n\alpha^2+\frac{n^2\alpha^2}{T-n}]
    \end{equation}  
\end{lemma}
\begin{proof}
    See Appendix \ref{app: proofs}.
\end{proof}
\noindent Now that we have the expected cash flow for the case when there are only two trading velocities we can compute which are the optimal ones that maximize the cash flow. The following lemma proves this result.
\begin{lemma}\label{lem:second lemma}
    Assume the conditions of Lemma \ref{lem: expected cash} hold and $\alpha>0$, then the maximum expected cash flow at the end of the trading period is
    \begin{equation}\label{eq: PnD}
        \max_{n\in[T-1]}\E[X_T](\alpha^*,n)~~~~~\text{with}~~~~\alpha^* = \bigg(\frac{(\delta+1)A(n)}{2B(n)}\bigg)^{\frac{1}{1-\delta}}
    \end{equation}
    and
    \begin{equation*}
        A(n) = \frac{n\theta}{2}\Big[(n+1)-(T-n-1)\Big(\frac{n}{T-n}\Big)^{\delta}\Big]~~~~~,~~~~~~
        B(n) = \kappa\Big[n+\frac{n^2}{T-n}\Big].
    \end{equation*}
\end{lemma}
\begin{proof}
    See Appendix \ref{app: Frame the problem}
\end{proof}
\noindent

\begin{wrapfigure}[16]{r}{0.48\textwidth}
    \centering
    \includegraphics[width=\linewidth]{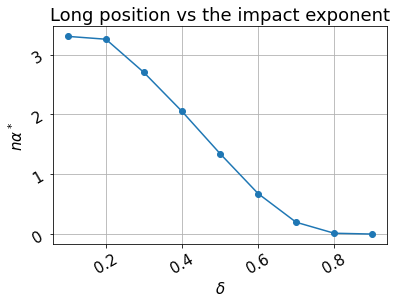}
    \caption{The optimal long position, \(n\alpha^*\), as a function of
    \(\delta\) for the PnD strategy.}
    \label{fig:long-vs-delta-pnd}
\end{wrapfigure}

Notice that we do not have a closed-form analytical solution for the optimal value of $n$, but only for the optimal value of $\alpha$, which itself depends on $n$. Therefore, when $|\mathcal{A}|=2$, we first determine the optimal trading velocity $\alpha^*$ for each fixed $n$, and then choose the value of $n$ that maximizes the expected cash flow. Finally, for a fixed $n$, there exists a unique strategy that maximizes the expected cash flow when $\alpha>0$. Throughout the remainder of the paper, we refer to this strategy as the pump-and-dump strategy (PnD).

It is important to stress that nonlinearity is the main source of manipulation.
Fig.~(\ref{fig:long-vs-delta-pnd}) shows that the long position in the optimal
PnD strategy decreases as $delta$ increases. This relationship illustrates
how stronger nonlinear price impact increases the size of the position selected
by the optimal manipulation strategy.




\subsection{Numerical optimization via SLSQP}\label{subsec: Strategy by SLSQP}

Assuming only two distinct trading velocities, as in subsection \ref{subsec:twoVelo}, clearly restricts the space of possible strategies and the revenues generated from trading. A richer set of trading volumes and strategies could emerge if $|\mathcal{A}|$ were allowed to include more than just two velocities. However, in that case the problem is not convex any more and there are no closed-form solutions. We therefore resort to numerical methods to obtain an approximate solution to the problem. We solve the constrained optimization problem using the Sequential Least Squares Programming (SLSQP) algorithm proposed by \cite{kraft1988software}. Therefore, the resulting optimization problem reads as follows:
\begin{align*}
    \vec{v}_* = \argmax_{\vec{v}\in\R^T}\E[X_T],~~~~~s.t~~\sum_{t=1}^Tv_t = 0,
\end{align*}
The objective function $\mathbb{E}[X_T]$ is defined in Eq.(\ref{eq: final cash-flow_main text}), and the parameter values used to determine it are reported on Table~\ref{tab:params RL and NS}.
We initialize the optimization from $N=500$ different starting points, sampling uniformly at random from $U[-b,b]^T$, in order to identify the best solution.
\begin{table}[ht]
\centering
\caption{Fixed Parameters used both for the RL environment and for the definition of the objective for the SLSQP.}
\label{tab:params RL and NS}
\begin{tabular}{|c|c|c|}
\hline
\textbf{Description} & \textbf{Symbol} & \textbf{Value} \\
\hline
Time Horizon & $T$ & 21 \\
Permanent Impact Exponent & $\delta$ & 0.1 \\
Permanent Impact Coefficient & $\theta$ & $0.0001$ \\
Temporary Impact Coefficient & $\kappa$ & $0.0003$ \\
Maximum Buy (Sell) & $b$ & $15$ ($-15$)\\
\hline
\end{tabular}
\end{table}
Clearly, in this first step, to use the numerical solution we assume complete knowledge of the parameters governing the evolution of the trading environment. 

\begin{figure}[t]
    \centering
    \subfigure[]{\includegraphics[scale=0.4]{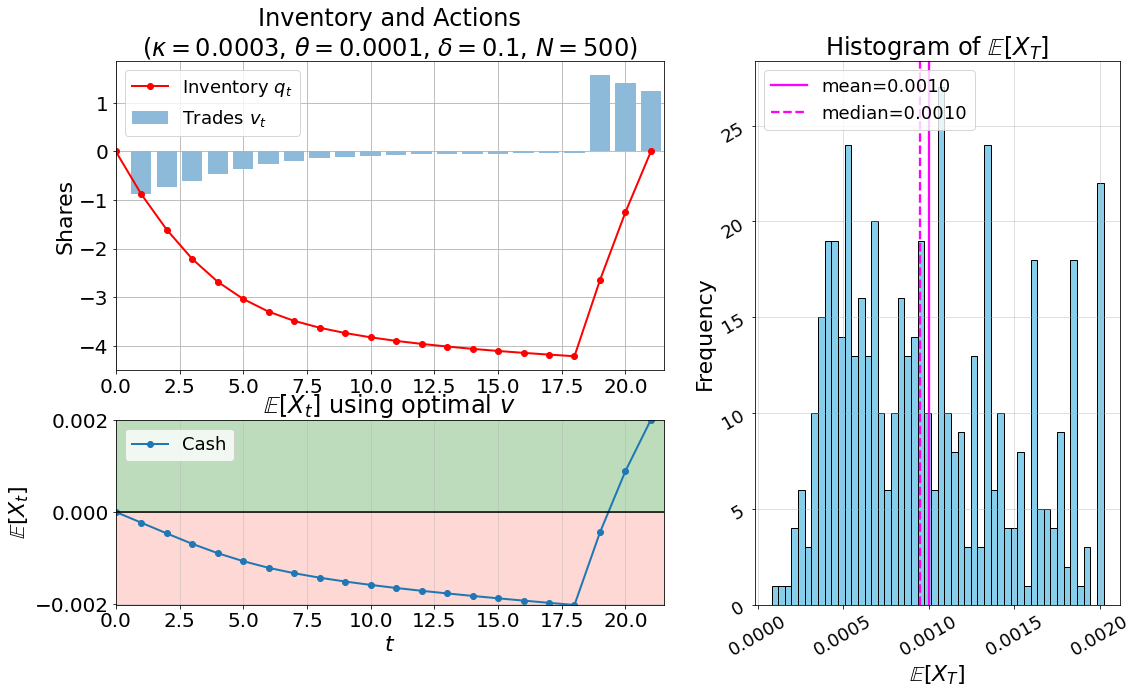}}
    \caption{{\bf Top-Left:} The optimal trading strategy by SLSQP and the corresponding inventory (red line). {\bf Bottom-Left:} The expected cash-flow. The green (red) area indicates the region for which $\E[X_t]>0$ ($\E[X_t]<0$). {\bf Right:} The histogram of $\E[X_T]$ from the different starting points. {\bf Note:} The parameters used are the ones reported on Table \ref{tab:params RL and NS}.}
    \label{fig:numSol}
\end{figure}

Fig.~(\ref{fig:numSol}) illustrates the best numerical solution obtained via SLSQP. The top-left panel shows the optimal trading velocity together with the corresponding inventory trajectory, depicted by the red line. The bottom-left panel reports the associated expected cash flow, from which we observe the existence of a manipulative strategy since $\mathbb{E}[X_T] > 0$. The right panel displays the histogram of $\mathbb{E}[X_T]$ obtained from 500 different SLSQP initializations. These results highlight the complexity of the optimization problem, which stems from the non-linear form of the permanent market impact function. In Appendix \ref{subsec: Similarities among Trading Strategies} we shed light on the landscape of the various trading strategies found by the SLSQP starting from the $N$ different initial points.

\subsection{Effect of impact parameter estimation on profitability}\label{subsec: Misspecification of Parameters in SLSQP}
The previous numerical results were based on the assumption that the parameters governing the market were fully known when optimizing with SLSQP. In practice, this corresponds to an idealized setting that cannot realistically be solved under full information. In Section \ref{subsec: Sensitivity of PnL to parameter misspecification}, we study the sensitivity of SLSQP to parameter misspecification, while in Section \ref{subsec: Estimate parameters SLSQP}, we present a method of impact parameters estimation on a finite sample of optimal execution strategies, based on \cite{Almgren_Thum_Hauptmann_Li}.

\subsubsection{Sensitivity of PnL to parameter misspecification}\label{subsec: Sensitivity of PnL to parameter misspecification}

 Usually, most of the times in real market scenarios, impact parameters values are not available, as by definition they are latent quantities and therefore difficult to be estimated directly.  Therefore, numerical solutions obtained with estimated parameters might be affected by estimation errors as magnitudes might be misspecified. To see this, 
Fig.~(\ref{fig:cash_vs_delta_overall}) illustrates the expected cash as a function of parameter misspecification across all three parameters. For the computation of the cash flow, we find the optimal strategy by using $\hat{\delta}$ and then we substitute this strategy to Eq.(\ref{eq: final cash-flow_main text}) where we use the parameter $\delta$, similarly for the other two cases when we do not have knowledge of parameters $\theta$ and $\kappa$. Specifically, panel (a) shows the impact of misspecification in $\delta$, panel (b) in $\theta$, and panel (c) in $\kappa$.


Two key observations emerge. First, clearly the highest expected cash is achieved in the absence of misspecification, and performance deteriorates as the magnitude of the misspecification increases. Second, performance improves when the true value of $\delta$ is closer to 1 rather than to 0. A similar pattern is observed in the remaining panels, where the highest cash is again attained under no or minimal misspecification.
\begin{figure}[h!]
    \centering
    \subfigure[]{
    \includegraphics[width=0.32\linewidth]{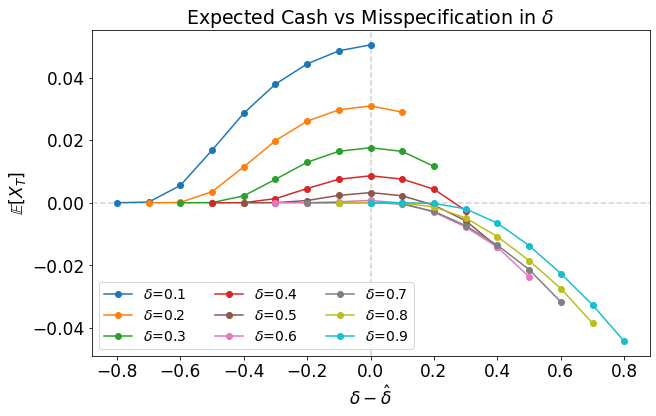}}
    \subfigure[]{\includegraphics[width=0.32\linewidth]{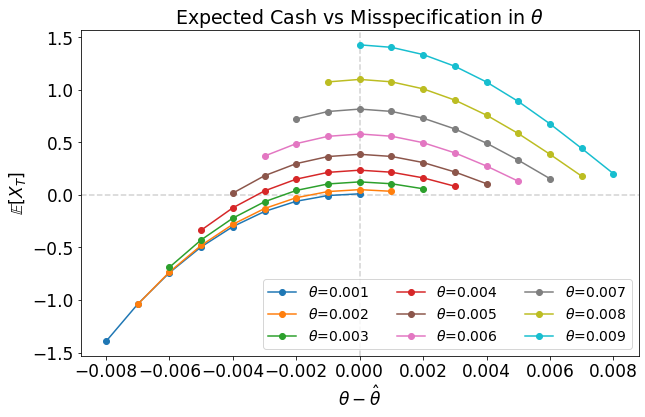}}
    \subfigure[]{\includegraphics[width=0.32\linewidth]{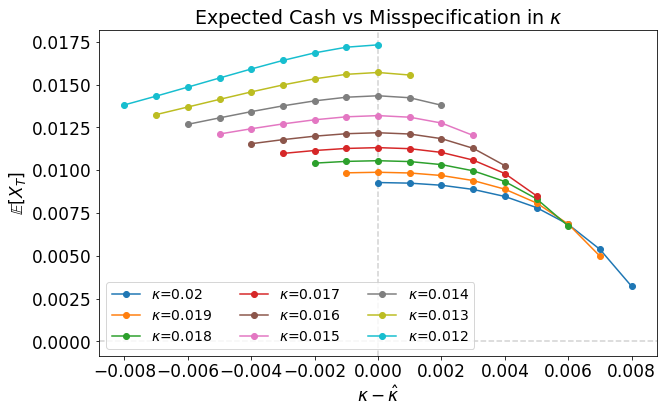}}
    \caption{Expected cash vs the true parameters; $\delta,\theta$ and $\kappa$ when the optimal strategy is found using $\hat{\delta},\hat{\theta}$ and $\hat{\kappa}$.}
    \label{fig:cash_vs_delta_overall}
\end{figure}
\subsubsection{Estimation of impact parameters}\label{subsec: Estimate parameters SLSQP}
Here we present a parameter estimation method based on
\cite{Almgren_Thum_Hauptmann_Li}.  In \cite{Almgren_Thum_Hauptmann_Li}, the authors used real metaorders\footnote{Metaorders are large trading orders that are split into pieces and executed incrementally.}, thus here we estimate the parameters by \textit{simulating} metaorders from a TWAP strategy in a time window $[1,T]$ leaving the dynamics of the market unchanged from previous numerical experiments. Specifically, we assume that for the metaorder $i$ the fundamental and the execution price, respectively, are
\begin{align*}
    S_t^{(i)} & = S_{t-1}^{(i)}+\theta\text{sign}(v^{(i)})|v^{(i)}|^{\delta}+\sigma\xi_t^{(i)}\\
    \Tilde{S}_t^{(i)}& = S_{t-1}^{(i)}+\kappa v^{(i)},~~~~~t=1,\cdots,T~~\text{and}~~i=1,\cdots,N.
\end{align*}
We assume that at the $i$-th metaorder we want to execute $q_i$ shares\footnote{In \cite{Almgren_Thum_Hauptmann_Li} this variable is denoted by $X$.} at a constant speed 
for any order $i$ and it is $v^{(i)} = \frac{q_i}{T}$ with $q_i \sim U[1,20]$.

For each metaorder $i$ the fundamental price at time $t$ can be written as:
\begin{equation*}
    S_{t}^{(i)} = S_0 + t\theta(v^{(i)})^{\delta}+\sigma\sum_{s=1}^{t}\xi_{s}^{(i)}
\end{equation*}
thus, at the end of metaorder $i$, at time $T$ the fundamental price it is
\begin{align*}
    S_T^{(i)} & = S_0+T\theta (v^{(i)})^{\delta}+\sigma\sum_{s=1}^{T}\xi_{s}^{(i)}\Leftrightarrow\\
    \underbrace{\frac{S_T^{(i)}-S_0}{S_0}}_{I_i} & = \frac{T\theta(v^{(i)})^{\delta}}{S_0}+\text{noise}
\end{align*}
The execution price at time $t$ can be written as:
\begin{equation*}
    \Tilde{S}_t^{(i)} = S_0+(t-1)\theta (v^{(i)})^{\delta}+\sigma\sum_{s=1}^{t-1}\xi_{s}^{(i)}+\kappa v^{(i)}
\end{equation*}
and the realised price (following the terminology in \cite{Almgren_Thum_Hauptmann_Li}) for the metaorder $i$, due to the TWAP strategy, is 
\begin{equation*}
    \Bar{S}_i = \frac{\sum_{t=1}^Tv^{(i)}\Tilde{S}_{t}^{(i)}}{\sum_{t=1}^Tv^{(i)}} = \frac{\sum_{t=1}^T\frac{q_i}{T}\Tilde{S}_{t}^{(i)}}{\sum_{t=1}^T\frac{q_i}{T}} = \frac{\frac{1}{T}q_i\sum_{t=1}^T\Tilde{S}_{t}^{(i)}}{q_i}
\end{equation*}
\begin{align*}
    \Bar{S}_i & = \frac{1}{T}\sum_{t=1}^T \Tilde{S}_{t}^{(i)}\\
    & = \frac{1}{T}\sum_{t=1}^T\Big[S_0+(t-1)\theta (v^{(i)})^{\delta}+\kappa v^{(i)}+\sigma\sum_{s=1}^{t-1}\xi_{s}^{(i)}\Big]\\
    & = S_0+\theta \frac{T-1}{2}(v^{(i)})^{\delta}+\kappa v^{(i)}+\text{noise}\Leftrightarrow\\
    \underbrace{\frac{\Bar{S}_i-S_0}{S_0}}_{J_i} & = \frac{T-1}{2T}I_i+\frac{\kappa v^{(i)}}{S_0}+\text{noise}
\end{align*}
Therefore, we end up with the following two regressions\footnote{Notice that in \cite{Almgren_Thum_Hauptmann_Li} they use the approximation $\frac{T-1}{T}\approx 1$, to simplify the second regression equation but we observe, without this approximation, that the estimations for the temporary impact are more robust, since for $T=21$ the approximation is not sufficiently accurate.}:
\begin{align}
    I_i & = \frac{T\theta(v^{(i)})^{\delta}}{S_0}+\text{noise}\\
    J_i-\frac{T-1}{2T}I_i & = \frac{\kappa v^{(i)}}{S_0}+\text{noise}
\end{align}
To find an estimate for $\delta$ we minimize the residuals sum of squares (RSS): 
\begin{equation}
    \hat{\delta} = \argmin_{\delta\in(0,1)}\sum_{i=1}^N\Big(I_i-\hat{\theta}(\delta)\frac{T}{S_0}\Big(\frac{q_i}{T}\Big)^{\delta}\Big)^2
\end{equation}
where $\hat{\theta}(\delta)=\frac{S_0}{T}\frac{\sum_{i=1}^N(\frac{q_i}{T})^{\delta}I_i}{\sum_{i=1}^N(\frac{q_i}{T})^{2\delta}}$ (see Appendix \ref{app: Estimating the Parameters in SLSQP} for a derivation of this formula). Then $\hat{\theta} = \hat{\theta}(\hat{\delta})$. Similarly, for the estimation of $\kappa$ it is
\begin{equation*}
    \hat{\kappa} = S_0\frac{\sum_{i=1}^N(\frac{q_i}{T})(J_i-\frac{T-1}{2T}I_i)}{\sum_{i=1}^N(\frac{q_i}{T})^2}.
\end{equation*}

\subsection{Numerical experiments}

We now perform numerical experiments to show how the estimation method proposed above works in finite sample. The parameter of the impact model are those in Table \ref{tab:params RL and NS}. For the choice of the volatility $\sigma$, we adopt a realistic estimate by using real data. First we use an estimate of the daily realized volatility (RV) (see \cite{Real_Vol}) of MSFT from April 2018 and we set $\sigma_\text{daily} = 0.026$. Similarly to \cite{Almgren_Thum_Hauptmann_Li} we assume that an average metaorder lasts $0.39$ days, thus
\begin{equation*}
    \sigma_\text{metaorder}:=\sigma = \sigma_\text{daily}\sqrt{0.39}=0.0168.
\end{equation*}
Moreover for the impact parameters; $\theta$ and $\kappa$, we choose them in a way such that the ratios $\frac{\sigma}{\theta}$ and $\frac{\sigma}{\kappa}$ are 56 and 167 respectively. These ratios are the same with the ones reported on Table B from \cite{Almgren_Thum_Hauptmann_Li} in which they report the summary statistics from a sample of real metaorders. Finally, we consider the estimation on samples of size $N\in\{500,5000\}$.

The results are shown in Table \ref{tab: estimation of theta,kappa and delta}. As expected, the estimation is more precise for larger samples\footnote{We conducted experiments for larger values of $N$, which provide indications that the estimators are asymptotically consistent.} and for smaller volatility. The interesting aspect is that the estimation of the parameter $\delta$ is  very noisy and usually upward biased. As suggested by Fig.~(\ref{fig:long-vs-delta-pnd}), this overestimation leads to more conservative strategies (i.e. with smaller long position), reducing the potential PnL of the dynamic arbitrage strategy when the DGP is known but the impact parameters must be estimated from a finite sample of data. 

\begin{table}[]
    \centering
    \small
    \begin{tabular}{c| p{.6cm} p{.5cm} c c c c c || c c}
    
        Parameter & $(10^{-4})$ & $\Big[\hat{\theta}$ & SE of $\hat{\theta}$ & $\hat{\kappa}$ & SE of $\hat{\kappa}\Big]$ & $\hat{\delta}$ & SE of $\hat{\delta}$ & $\sigma$ & $N$\\[5pt]\hline
        \multirow{7}{*}{Estimate} & &$0.61$ & $3$ & $2$ &  $17$ & $0.99$ & $1.88$ & $0.0168$ & \multirow{3}{*}{500}
        \\
        & &$1$ & $0.48$ & $4$ & $3$ & $0.79$ & $0.94$ & $0.0030$
        \\
        & &$1$ & $0.23$ & $1$ & $2$ & $0.20$ & $0.31$ & $0.0020$
        \\[10pt]
        
        & &$1$ & $0.63$ & $9$ &  $5$ & $0.26$ & $0.61$ & $0.0168$ & \multirow{3}{*}{5,000}
        \\
        & &$1$ & $0.12$ & $3$ & $0.09$ & $0.48$ & $0.22$ & $0.0030$
        \\
        & &$1$ & $0.07$ & $3$ & $0.65$ & $0.15$ & $0.10$ & $0.0020$
         \\\hline
         True & &$1$ & & $3$& &$0.1$&\\
        \hline
    \end{tabular}
    \caption{Estimate and Standard Error (SE) of the parameters when $T=21$. Note that the values for the first 4 columns are in basis points.}
    \label{tab: estimation of theta,kappa and delta}
\end{table}

\section{Dynamic Arbitrage via deep RL}
\label{sec:DDPG}
Having set the stage with numerical techniques, we now pass to studying the problem with deep RL. By means of RL we aim at assessing whether a learner, who has no direct knowledge of either the impacts structure or of the values for the impacts, would be able to \textit{learn} a dynamic arbitrage strategy, employing as little information on the market as possible. In a spirit similar to that of \cite{Ning_2021} and \cite{Macri_2025}, we solve the round trip optimal trading problem in the context of the model reported in Eq.\eqref{eq:fund price} with DDPG.   

\subsection{Environment, states and reward}

The deep RL experiment runs over a number $M$ of training episodes of trading where the agent learns the optimal strategy. The agent, for each time step $t$, has access to a simple state representation of the environment called $\state = (q_t, t, S_t)$, that is the tuple made by current inventory at time step $t$, time step $t$ and permanently impacted price at time step $t$, i.e. $S_t$. 

The price dynamics evolve as in Eq.\eqref{eq:fund price} and the values of the parameters are the same as the ones described on Table \ref{tab:params RL and NS}, therefore the agent has to make a trading decision $v_t$ (i.e. buy/sell action to adopt) based on the knowledge of the state $\state$, such that the reward at each time step is maximal. The reward, seen at each trading step is defined as:
\begin{equation}
\label{eq:reward}
   r_t = S_t v_t-\kappa v_t^2
\end{equation}
Where $v_t$ is the traded volume at time $t$, $S_t$ is the permanently impacted price as in Eq. \eqref{eq:fund price} and $\kappa$ is the temporary impact parameter. The reward and the actions are added to the tuple $\state$ when they are stored in the memory used by the algorithm to train its neural networks (further details below).

\subsection{The learning algorithm}

 The DDPG algorithm is composed by two main neural networks: a \textit{critic} and an \textit{actor} network. It allows for continuous controls and combines deterministic policy gradients with deep function approximation. Originally introduced by \cite{pmlr-v14-lillicrap}, DDPG extends the deterministic policy gradient framework of \cite{pmlr-v32-silver14} by leveraging neural networks to approximate both the policy (actor) and the action-value function (critic). The actor network deterministically maps states to actions, while the critic estimates the expected return. To stabilize training, DDPG employs target networks and an experience replay buffer, borrowing key ideas from Deep Q-Network. The critic is trained by minimizing the Bellman error using targets constructed from the reward signal and the discounted estimate of future returns, while the actor is updated via the deterministic policy gradient, maximizing the critic’s evaluation of the chosen actions.

\subsubsection{The actor and critic networks}
The critic and actor neural networks are feed-forward networks with four and three input neurons, respectively, a number of $l_{\text{NN}}$ layers of $d_{\text{NN}}$ hidden nodes, each with \verb!SiLu! activation function, in the Actor network, the final layer has \verb!tanh! activation function. 

We refer to the actor network as $\pi$, since it is the network that parametrises the policy directly, and to the critic network as $Q$, as the architecture of the critic part of the algorithm relies fully on double deep Q-learning.



At the beginning of the training loop we initialise the networks $\pi, Q$ with random weights $\mu_\pi, \mu_Q$. As training unfolds, $(\mu_\pi, \mu_Q) \to (\mu^*_\pi, \mu^*_Q)$, where $(\mu^*_\pi, \mu^*_Q)$ are the optimal weights that correspond to the optimal policy that maximises the the sum of discounted rewards. 
 
Therefore, during training the $Q$ network, in order to avoid an overestimation of the $Q$-values for the considered state-action pairs, we use a $Q_{\text{tgt}}$ network that is initialised at the beginning of the training routine as a copy of the $Q$ network with weights $\mu_{Q_{\text{tgt}}} = \mu_{Q}$, thus employing double deep Q-learning for the Critic part of the algorithm, the weights of the $\mu_{Q_{\text{tgt}}}$ are then periodically set equal to $\mu_Q$. In order to do so, we have adopted the soft update technique as described in the original DDPG paper by \cite{pmlr-v14-lillicrap}, with the soft update parameter set to $0.001$.

During training, as the trading unfolds we chose to keep memory of the states seen and the actions, volumes traded, at each time step into a memory of length $m$. As soon as the memory has at least length $B$, which is the length of a batch, the algorithm starts training the two networks via gradient descent. In any case, every time the networks are trained, batches of past trading experiences are randomly sampled from the memory.

\paragraph{The critic.}

For each training iteration, we train a critic network $Q$. To this end, we add an additive exploration zero mean noise to the output of the actor network $\pi$ (whose training is described below), whose variance declines at each iteration\footnote{Specifically, before learning loop starts we set $a=1/T$ and $\varepsilon_{\text{min}} = 0.1$ 
Then, for each train iteration $m$ we set $\varepsilon = \max(\varepsilon - a, \varepsilon_{\text{min}})$. In this way, the more we iterate in the learning routine the lower the $\varepsilon$ value. $\varepsilon$ is then used as an additive noise in the Actor-Critic architecture used in both the learning algorithms.}.
The additive noise, $\varepsilon$, helps the Actor neural network to explore the range of inventory holdings to be held depending on the state.
Thus, while training the $Q$ network, $v_{{t+1}} = \pi(\state|\mu_\pi) + \mathcal{N}(0, \varepsilon)$, where \(\pi(\state|\mu_\pi)\) is the output of the Actor network. Clearly, as the training unfolds and $\varepsilon$ becomes smaller, the exploration noise diminishes and the policy chosen by the network $\pi$ becomes deterministic. In this way, given the weights $\mu^*_\pi$, found at the end of the training, the Critic network directly maps a state into an action such that:
\begin{equation}
    v^{*}_t = \pi(\state|\mu^*_\pi)
\end{equation}
Once the volume is obtained, it is used to calculate the reward for each time-step using the equation Eq.\eqref{eq:reward}.

To train the $Q$ network, we update the weights $\mu_Q$ such that $\mu_Q = \argmin_{\mu_Q} \L_1(\mu_Q; \mu_{Q_{\text{tgt}}})$, by taking a single gradient step on the loss 

\begin{align}
        \L_1 &= \frac{1}{b} \sum^b_{i=1} (Q(\state, v_{{t+1}}|\mu_Q) - y_t)^2\\
        y_t &= r_{t} + \gamma Q_{\text{tgt}}(\state', \pi(\state'|\mu_\pi)|\mu_{Q_{\text{tgt}}})
\end{align}

Where $\state '$ is the next state where the environment evolves after the agent has taken the action $v_t$.
The weights of the Critic network $\mu_Q$ are then updated using gradient descent.

\paragraph{The actor.} 

For each training iteration we train the actor network $\pi$, which is the neural network responsible for choosing the action $v_t$. 
To do this, we feed the $\pi$ network with a batch of $B$ states $\state$.  The $\pi$ network returns a volume to be sold or bought $v_t$, which is in turn fed to the Critic $Q$ network, along with the state $\state$. The output of the $Q$ network is the $Q$ value which is maximized during training. Specifically, the weights $\mu_\pi$ of the Actor network are found by maximizing the output of the critic network with fixed weights $\mu_Q$. As before we perform a single gradient step to minimise the loss:

\begin{equation}
    \L_2 = - \frac{1}{b} \sum^b_{i =1} Q(\state, \pi(\state'|\mu_\pi)|\mu_Q)
\end{equation}
where $\pi(\state|\mu_\pi)$ is the output of the actor network with current weights $\mu_\pi$.
The weights $\mu_\pi$ are updated using the policy gradient theorem. The derivative of the loss function $\L_2$ with respect to the weights of the Actor network is:
\begin{equation}
    \nabla_{\mu_\pi}\L_2 = - \frac1b\sum^b_{i=1}\left[ \nabla_{a}Q(\state, a^{(i)}|\mu_Q)|_{a^{(i)} = \pi(\state'|\mu_\pi)} \nabla_{\mu_\pi}(\state|\mu_\pi) \right]
\end{equation}
where the first part in the square brackets tells us how sensitive the $Q$ network is to changes in the action chosen by the Actor $\pi$, this is the gradient of the $Q$-value with respect to the action and evaluated at the action suggested by the Actor network, in this context for ease of notation we use $a^{(i)}=v^{(i)}_t$ to denote the action for the \(i\) state among the $b$ considered in the batch. The second part inside the brackets is the gradient of the the actor's policy and accounts for how much the action $v_t$ changes with respect to changes in the $\pi$ network parameters. We update the weights for the $\pi$ network with gradient descent. We repeat this training routine a number $l$ of times in order to assess how the quality of the actions chosen by the Actor changes and, in this way, to thoroughly explore weights combinations for the $\pi$ network. 

The complete training algorithm for both the Critic and the Actor neural networks is reported in Algorithm~\ref{alg:ddpg}. Both the Actor and the Critic neural networks are feed-forward fully connected neural networks, for the gradient descent we use the Weighted ADAM with learning rate of $0.1$. As training unfolds, we adopt a scheduler for the learning rate, thus diminishing its value as the agent makes experience and the neural networks are trained, once very $100$ episodic iterations (i.e. time step over training episodes) the learning rate is diminished following the rule $lr \leftarrow lr \times 0.99$.
The features for the DDPG algorithm are normalised in the domain $[-1, 1]$. Relevant parameters for the networks and the market model are reported in Table~\ref{tab:params_RL_and_ABM}.

\begin{table}[h]
\centering
\caption{Symbol and value used for each parameter in the DDPG algorithm.}
\begin{tabular}{|c|c|c|c|}
\hline
Description &  Symbol & Value \\
\hline
Training Episodes & $M$ & 500\\
Memory Length & $m$ & 10e5\\
Batch Size & $B$ & 128\\
Learning Rate actor/critic & $\eta_{\pi} = \eta_Q$ & 0.1\\
Nodes actor/critic & $N_{\pi}=N_{Q}$ & 64\\ 
Layers actor/critic & $L_{\pi}=L_{Q}$ & 3\\ 
Discount Factor & $\gamma$ & 0.99\\
\hline
\end{tabular}
\label{tab:params_RL_and_ABM}
\end{table}
\begin{figure}
    \centering
    \includegraphics[width=1\linewidth]{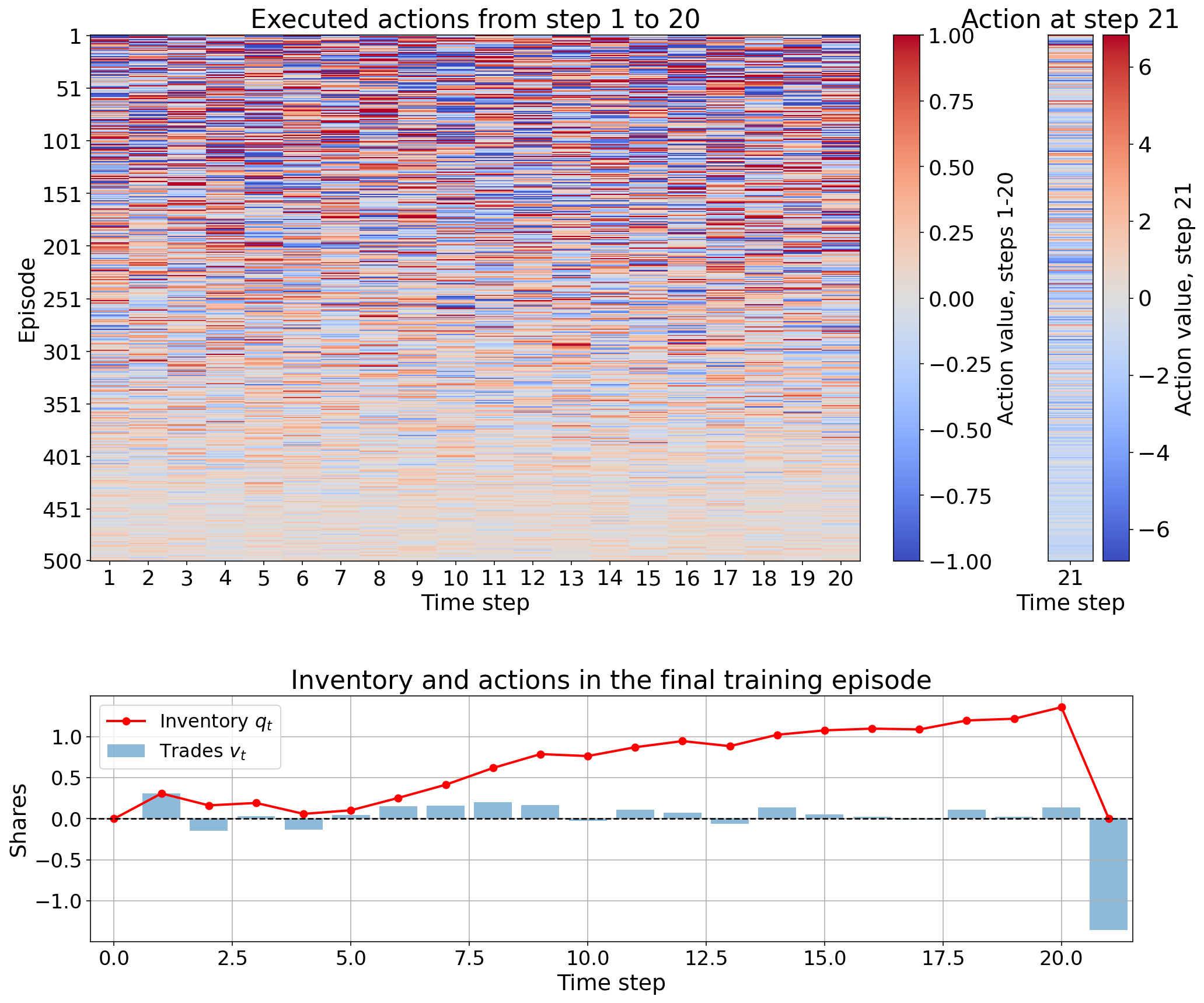}
    \caption{{\bf Top:} Heatmap of the executed actions chosen by DDPG during the training phase of 500 episodes. {\bf Bottom:} The actions chosen from DDPG along with the corresponding inventory during the final training episode.}
    \label{fig:Heatmap in training - 500episodes}
\end{figure}

\begin{algorithm}[ht]
\caption{Training of DDPG agent for optimal trading}
\begin{algorithmic}
\Require
{\\

Initialise with random weights $Q$ and make a copy $Q_{\text{tgt}}$;\\
Initialise with random weights $\pi$\\
Initialise the memory with max length $m$. Set $\varepsilon = 1$, $B$ batch size, $M$ train episodes;\\
Set market dynamics parameters;\\
}
\For {i in M}
    \State Set $S^i_0 = S_0$;\\
    \For{t in T}
        \State {$\state \gets (q_t,t, S^i_t)$;}
        \State {$v_t =  \pi(\state|\mu_\pi);$}\\
        \State{$r_t\gets S^i_{t-1} v_t - \kappa v^2_t$;}\\
        \State{$S^i_{t-1}\to S^i_t$;}\Comment{Generate $S^i_t$ from $S^i_{t-1}$}\\
        \State{$\state' \gets (q_{t+1}, t+1, S_{t+1})$ ;}\\
        \State {Memory $\gets(\state,r_{t},v_t, \state')$;} \Comment{Memory storing}
        \\
        \If{Length of memory $\ge$ $B$}
            \For{j in $B$}
                \State{Sample a batch of $(\state^j,r^j_{t},v^j_t,\state^{', j})$ from memory;}
                \State train $\pi$;
                \State train $Q$
            \EndFor\\
            \If{Length of memory = $L$} halve the length of memory\EndIf
        \EndIf
        \State{After $m$ iterations decay $\epsilon=\epsilon - a$;}
        \State{After $m$ iterations $\theta_{\text{tgt}}\gets\theta_{\text{main}}$;}
    \EndFor
\EndFor
\end{algorithmic}
\label{alg:ddpg}
\end{algorithm}
\begin{figure}[h!]
    \centering
 \includegraphics[width=1\linewidth]{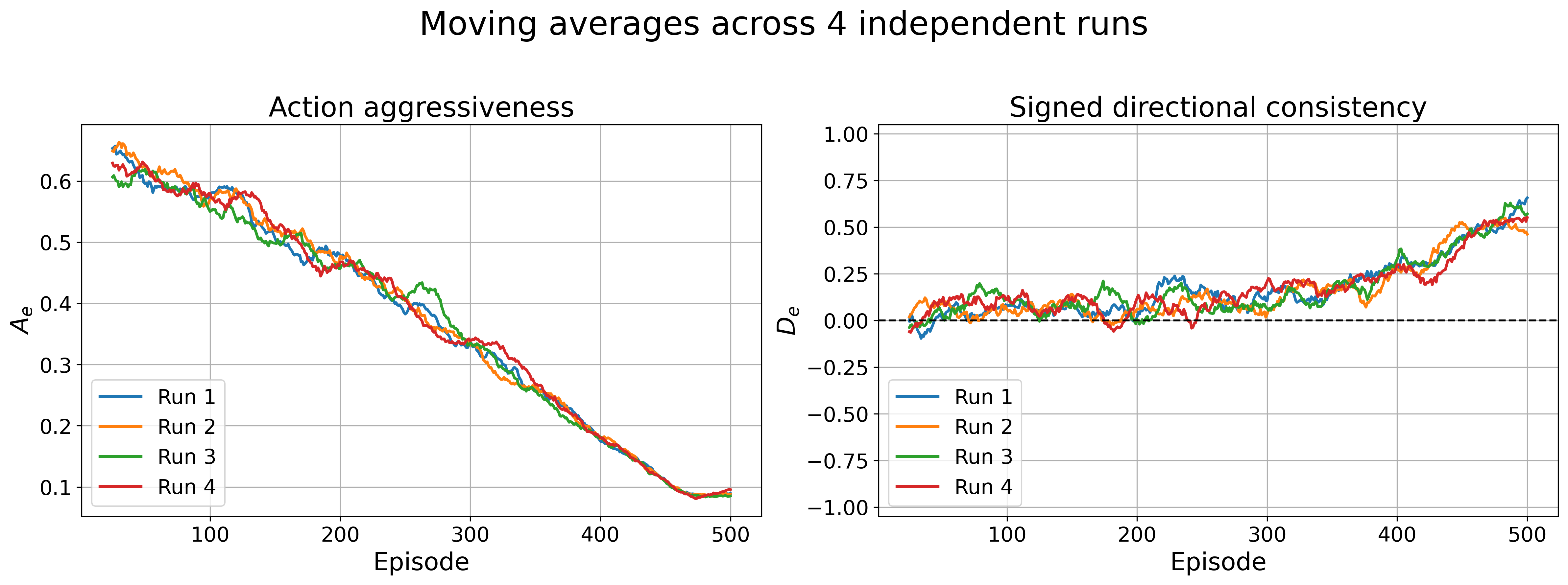}
    \caption{The 25-episode moving averages {\bf Left:} of the action aggressiveness and {\bf Right:} of signed action directionality during training.}
    \label{fig:Aggresiveness and Directionality}
\end{figure}
\paragraph{Training results.} 
We train the agent using the procedure described above and terminate training after episode 500, since no noticeable improvement in the learned policy is observed beyond this point. The top panel of Fig.~(\ref{fig:Heatmap in training - 500episodes}) shows a heatmap of the actions selected at each time step and in each episode during one run of training. At the beginning of the training phase, the agent alternates between exploratory buy and sell actions. As training progresses, however, the policy converges toward a directional buying strategy over the first 20 time steps, followed by a sell execution at the final step to satisfy the round-trip constraint. The bottom panel shows the final learned strategy.

To characterize how the strategies learned by DDPG evolve across episodes, we define the action aggressiveness in episode $e$ as:
\begin{equation}\label{eq: aggressiveness}
    A_e = \frac{1}{20}\sum_{t=1}^{20}|v_{e,t}|
\end{equation}
and the signed directionality in episode $e$ as:
\begin{equation}\label{eq: directionality}
    D_e = \frac{\sum_{t=1}^{20} v_{e,t}}{\sum_{t=1}^{20} |v_{e,t}|},~~~~~\text{with}~~-1<D_e<1
\end{equation}
where $v_{e,t}$ is the action chosen (i.e. the volume traded) in episode $e$ at time $t$. The quantity $A_e$ measures how aggressive in absolute terms is the pump phase of the manipulative strategy, while $D_e$ tells us how consistently the strategy buys or sells in the pump phase. The left panel of Fig.~(\ref{fig:Aggresiveness and Directionality}) displays the evolution of $A_e$ during 4 training runs. We observe that, as training progresses, DDPG converges toward less aggressive strategies. The right panel of Fig.~(\ref{fig:Aggresiveness and Directionality}) shows $D_e$ across the same runs. During the first episodes, this quantity is close to zero, indicating no clear directional preference in the first 20 time steps. As training proceeds, however, the strategy converges toward a positive direction over the first 20 time steps. Notice that this symmetry breaking is consistent with the form of the reward function in Eq. (\ref{eq:reward}). At each time step $t$, DDPG seeks an action $v_t$ that maximizes the instantaneous reward. Therefore, by differentiating the reward with respect to $v_t$ and setting the derivative equal to zero, we obtain the action that maximizes the reward at time $t$:
\begin{equation}
    v_t^* = \frac{S_t}{2\kappa}>0.
\end{equation}
If on the other hand, one defines the reward as:
\begin{equation*}
    r_t = -S_tv_t-\kappa v_t^2
\end{equation*}
then we will observe a negative direction at the final episodes (see Fig.~(\ref{fig:Aggr vs Dirct}) in Appendix \ref{app: training DDPG}).





\section{Discovering price manipulation strategies in finite data samples}
\label{sec:comparison}

We now consider the problem of a trader who tries to find a price manipulation strategy in a market where dynamic arbitrage is allowed by the nonlinearity of permanent impact. Critically, we assume that the trader has a finite amount of data to identify the strategy, and we compare two approaches.

In the first one, the agent has $N$ TWAP metaorders to estimate the parameters of the AC nonlinear model with the regression method presented above. Then the agent uses the estimated parameters in an SLSQP optimizer to find the optimal price manipulation and it tests it on 1000 out of sample simulations of manipulative strategies. We consider $N=500$ and $N=5000$ and we term these strategies $\text{regSLSQP}_{500}$ and $\text{regSLSQP}_{5000}$, respectively. In the second approach, the agent is agnostic about the underlying impact model and uses the RL approach described above to learn the manipulation strategy on $N=500$ episodes. As before, the agent tests it on 1000 out of sample simulations of manipulative strategies. We term this strategy DDPG.

As benchmarks we use the PnL obtained (i) with the optimal strategy with only two trading velocities and (ii) with the optimal strategy obtained with SLSQP. Clearly, in both cases we use the correct DGP and the true parameters. The two strategies are termed PnD and SLSQP, respectively. As above, we consider three levels of volatility.


Table \ref{tab:final comparison} summarizes, for each method, the average terminal cash $X_T$, the corresponding standard error, and the results of the t-test together with the associated p-values over 1000 out of sample implementations of the learned strategies.
For the regSLSQP methods, we follow a different approach. Since there is a large variability in the estimated impact parameters from a finite number of TWAP executions, the PnL in the out of sample set depends strongly on the specific realization of 500 or 5000 TWAPs used in the estimation. To mitigate this effect, we replicate the whole procedure 100 times and we report the mean and the average standard deviation in these 100 experiments. Moreover, the t-test column reports the percentage of the replications (among the 100) that reject the null-hypothesis of zero expected cash.


The results show that when volatility is high ($\sigma = 0.0168$), none of the methods among PnD, SLSQP and DDPG produces a statistically significant positive terminal cash. This is due to the finiteness of the test set and the large volatility. The regSLSQPs produce a statistically significant negative expected terminal cash in $46\%$ and $33\%$ of the cases when $500$ and $5000$ metaorders are used, respectively, for the estimation of the impact parameters. Thus in these conditions, even if manipulation is possible, large volatility prevents its exploitation.

For intermediate volatility values ($\sigma = 0.003$)  PnD, $\text{regSLSQP}{5000}$, and DDPG achieve statistically significant positive cash flows. PnD is the best-performing method followed by DDPG and then by $\text{regSLSQP}{5000}$. Moreover, SLSQP even though it has complete knowledge of both the functional form of the market-impact model and of the parameters does not produce statistically significant positive cash. On the contrary, the model based approach $\text{regSLSQP}{500}$ estimated on a smaller sample provides negative expected cash flow. Importantly, notice that the same sample size is sufficient for the RL based approach DDPG to learn a systematically profitable strategy.

Finally, when volatility is low ($\sigma = 0.002$), all methods except $\text{regSLSQP}_{500}$ generate statistically significant positive terminal cash. However, among the two sample based methods, $\text{regSLSQP}_{500}$ outperforms DDPG. This is somewhat expected because with low volatility and relatively large estimation set ($N=5000$), the parameters of the impact model are accurately estimated, while the DDPG is still trained on $N=500$ episodes.

\begin{table}[h!]
    \centering
    \begin{tabular}{c|c c c c|c}
        Method &  Mean $(10^{-4})$ & Std & t-test & p-value   & $\sigma$\\[5pt]\hline\hline
        PnD & $32$ & $0.1442$ & $0.70$ & $0.480$ & \multirow{5}{*}{$0.0168$}\\
        SLSQP & $48$ & $0.2805$ & $0.54$ & $0.580$\\
         \cline{1-5}
        \cellcolor{gray!20}$\text{regSLSQP}_{500}$ & \cellcolor{gray!20}$-5942$ & \cellcolor{gray!20}$1.5924$ & \cellcolor{gray!20}$46\%$ & \cellcolor{gray!20}$<0.05$\\
        \cellcolor{gray!20}$\text{regSLSQP}_{5000}$ & \cellcolor{gray!20}$-4082$ & \cellcolor{gray!20}$1.24$ & \cellcolor{gray!20}$33\%$ & \cellcolor{gray!20}$<0.05$\\
        \cellcolor{gray!20}DDPG & \cellcolor{gray!20}$11$ & \cellcolor{gray!20}$0.0434$ & \cellcolor{gray!20}$0.83$ & \cellcolor{gray!20}$0.400$\\ 
        \noalign{\vskip 2pt}
        \hline
        \noalign{\vskip 2pt}
        \hline
        \noalign{\vskip 2pt}
        PnD & $14$ & $0.0257$ & $1.75$ & $0.070$ & \multirow{5}{*}{$0.0030$}\\
        SLSQP & $25$ & $0.0501$ & $1.58$ & $0.110$\\
         \cline{1-5}
        \cellcolor{gray!20}$\text{regSLSQP}_{500}$ & \cellcolor{gray!20}$-2562$ & \cellcolor{gray!20}$0.1978$ & \cellcolor{gray!20}$34\%$ & \cellcolor{gray!20}$<0.05$\\
        \cellcolor{gray!20}$\text{regSLSQP}_{5000}$ & \cellcolor{gray!20}$1$ & \cellcolor{gray!20}$0.0561$ & \cellcolor{gray!20}$5\%$ & \cellcolor{gray!20}$<0.05$\\
        \cellcolor{gray!20}DDPG & \cellcolor{gray!20}$5$ & \cellcolor{gray!20}$0.0077$ & \cellcolor{gray!20}$2.34$ & \cellcolor{gray!20}$0.010$
        \\ 
        \noalign{\vskip 2pt}
        \hline
        \noalign{\vskip 2pt}
        \hline
        \noalign{\vskip 2pt}
        PnD & $13$ & $0.0171$ & $2.39$ & $0.010$ & \multirow{5}{*}{$0.0020$}\\
        SLSQP & $23$ & $0.0334$ & $2.22$ & $0.020$\\
         \cline{1-5}
        \cellcolor{gray!20}$\text{regSLSQP}_{500}$ & \cellcolor{gray!20}$-1071$ & \cellcolor{gray!20}$0.0852$ & \cellcolor{gray!20}$54\%$ & \cellcolor{gray!20}$<0.05$\\
        \cellcolor{gray!20}$\text{regSLSQP}_{5000}$ & \cellcolor{gray!20}$13$ & \cellcolor{gray!20}$0.0349$ & \cellcolor{gray!20}$36\%$ & \cellcolor{gray!20}$<0.05$\\
        \cellcolor{gray!20}DDPG & \cellcolor{gray!20}$5$ & \cellcolor{gray!20}$0.0051$ & \cellcolor{gray!20}$3.26$ & \cellcolor{gray!20}$0.001$\\\hline
    \end{tabular}
    \caption{Comparison among the methods after $1000$ simulations of the price evolution as in Eq.(\ref{eq:fund price})-(\ref{eq: exec price}).}
    \label{tab:final comparison}
\end{table}

As expected, volatility is the key parameter for the effective profitability of dynamic arbitrage strategies. When the model and the parameters are known (SLSQP and PnD), large volatility prevents the exploitation of the arbitrage, at least in a finite testing sample. When the model and/or the parameters are unknown, the volatility also enters in the estimation (regSLSQP) and learning (DDPG). The surprising result is that for intermediate volatility and small sample a model agnostic method is more effective of one knowing the model but not the parameter in identifying profitable strategies. Even more surprising is that in this volatility regime a ten-fold larger training set of the regSLSQP is not sufficient to beat the performance of DDPG. Finally, even for small volatility, DDPG outperforms regSLSQP when the training tests are equally small.  



In conclusion, the DDPG algorithm generates positive cash flows while relying on only minimal information: the current time step, the remaining inventory, and the permanently impacted asset price. Despite this simple state representation, the agent achieves performance comparable to numerical optimization methods, with relatively low variability in terminal cash $X_T$. This suggests that reinforcement learning algorithms can discover round-trip trading strategies that create dynamic arbitrage opportunities.

\section{Conclusion}

This paper studied whether a model-free reinforcement learning algorithm can learn to exploit price-manipulation opportunities more effectively than a traditional model-based approach when both methods are trained on finite samples of market data. We considered a single-asset optimal execution problem with concave permanent market impact, a setting in which dynamic arbitrage may arise through round-trip trading strategies.

We first established, both analytically and numerically, the existence of profitable manipulative strategies in discrete time. While these strategies can be derived under full knowledge of the market-impact model, implementing them in practice requires estimating latent impact parameters from data, introducing an additional source of uncertainty. This observation motivates the central comparison of the paper: a model-based optimization procedure with correctly specified dynamics but estimated parameters versus a model-free reinforcement learning agent that learns trading policies directly from interaction with the environment.

Our numerical experiments show that the DDPG agent is capable of learning profitable round-trip strategies without explicit knowledge of the underlying market-impact model. More importantly, when both approaches are constrained by the same amount of training data, reinforcement learning frequently outperforms the model-based procedure because it avoids the intermediate and error-prone parameter-estimation step. The advantage is most pronounced in regimes where parameter uncertainty materially affects optimization, whereas the model-based approach naturally regains superiority when estimation becomes sufficiently accurate.

The broader implication of these findings is methodological as much as financial. The paper does not argue that reinforcement learning creates manipulation where none exists; rather, it shows that, whenever the market environment admits exploitable dynamic arbitrage, a model-free learning algorithm can identify such opportunities directly from data and may do so more efficiently than a correctly specified parametric model estimated on finite samples. This highlights an important distinction between the theoretical existence of manipulation and its practical discoverability.

From a market-design perspective, the results also carry regulatory implications. Reinforcement learning agents optimize the objective they are given, without distinguishing between economically desirable profits and profits generated through behaviour that may undermine market integrity. Consequently, environments that inadvertently admit manipulative incentives may be particularly vulnerable to increasingly sophisticated learning algorithms. As reinforcement learning becomes more widely adopted in automated trading systems, ensuring that reward functions, market design, and regulatory constraints are aligned with no-manipulation principles will become increasingly important.

Several directions for future research naturally follow from this work. Extending the analysis to more realistic market-impact models, including transient or stochastic impact, would help assess the robustness of the results beyond the Almgren-Chriss framework. Likewise, studying richer market environments with limit-order-book dynamics, multiple strategic agents, or cross-asset interactions would provide a more comprehensive understanding of how learning algorithms exploit market structure. Finally, integrating explicit no-manipulation constraints or safe reinforcement learning techniques into optimal execution algorithms represents a promising avenue for developing trading systems that remain both profitable and market compliant.

\clearpage

\bibliographystyle{plainnat}
\begingroup
\let\emph\textnormal
\let\textit\textnormal
\let\textbf\textnormal
\bibliography{biblio_rl}
\endgroup

\appendix

\section{Derivations for section \ref{sec: Frame the problem}}\label{app: Frame the problem}
In this appendix we first derive the expected cash flow and then we give the proof of Proposition (\ref{prop1}).

We derive Eq.(\ref{eq: final cash-flow_main text})
by substituting the execution price into the expression for the terminal cash flow and expanding the result:
\begin{align*}
    X_T & = -\sum_{t=1}^T v_t\Tilde{S}_t\nonumber\\
    & = -\sum_{t=1}^T v_t(S_{t-1}+\kappa v_t)\nonumber\\
    & = -\Big[\sum_{t=1}^Tv_tS_{t-1}+\sum_{t=1}^T\kappa v_t^2\Big]
\end{align*}
The recursive Eq.(\ref{eq:fund price}) is written as 
\begin{equation*}
    S_{t-1} = S_0+\sum_{j=1}^{t-1}f(v_j)+\sigma\sum_{j=1}^{t-1}\xi_j
\end{equation*}
hence the terminal cash-flow can be written as:
\begin{equation*}
    X_T = -\sum_{t=1}^Tv_t\Big[\sum_{j=1}^{t-1}f(v_j)+\sigma\sum_{j=1}^{t-1}\xi_j\Big]-\kappa\sum_{t=1}^T v_t^2.
\end{equation*}
We observe that
\begin{equation*}
    \sum_{t=1}^T v_t\sum_{j=1}^{t-1} f(v_j) = \sum_{j=1}^{T-1}f(v_j)\sum_{t=j+1}^T v_t,~~~~~~v_0=0
\end{equation*}
thus it is
\begin{equation}\label{eq: Exp IS}
    \E[X_T] = -\Big[\sum_{t=1}^{T-1} f(v_t)\sum_{j=t+1}^Tv_j+\kappa\sum_{t=1}^T v_t^2\Big].
\end{equation}
The following proof establishes Proposition (\ref{prop1}).
\begin{proof}
Let the vector $\lambda v\in\R^T$ then the cost function is written as:
    \begin{align}\label{eq: well-posedeness}
        c(\lambda v) & = -\Big[\underbrace{\sum_{t=1}^{T-1} f(\lambda v_t)\sum_{j=t+1}^T\lambda v_j}_{F(\lambda v)}+\underbrace{\kappa\sum_{t=1}^T \lambda^2 v_t^2}_{S(\lambda v)}\Big]
    \end{align}
    the first term inside the parenthesis is written as
    \begin{align}
        F(\lambda v) & = \begin{bmatrix}
            f(\lambda v_1)\\
            f(\lambda v_2)\\
            \vdots\\
            f(\lambda v_{T-1})
        \end{bmatrix}^\intercal\begin{bmatrix}
            \sum_{t=2}^T\lambda v_t & 0 & 0 & 0 & \cdots & 0\\
            0 & \sum_{t=3}^T\lambda v_t & 0 & 0 & \cdots & 0\\
            \vdots &  \vdots &  \vdots &  \vdots & \cdots &  \vdots\\
            0 & 0 & 0 & 0 & \cdots & \lambda v_T
        \end{bmatrix}\begin{bmatrix}
            1 \\ 1 \\ \vdots \\ 1
        \end{bmatrix}\\
        & = |\lambda|^{\delta+1}F(v)
    \end{align}
    while for the second term it is clearly true that
    \begin{equation*}
        S(\lambda v) = \lambda^2 S(v).
    \end{equation*}
    hence
    \begin{equation*}
        c(\lambda v) = -[|\lambda|^{\delta+1}F(v)+\lambda^2 S(v)]
    \end{equation*}
    since $\delta<1$ the second term will dominate the convergence, thus the result in Eq.(\ref{eq: well-posedeness}) follows. For the continuity we simply observe that $c(v)$ is a sum of continuous functions.
\end{proof}
\section{Derivations for section \ref{subsec:twoVelo}}\label{app: proofs}
We first prove lemma \ref{lem: expected cash}
\begin{proof}
    \begin{align*}
        \sum_{t=1}^{T-1} f(v_t)\sum_{j=t+1}^Tv_j & = f(\alpha)(-\beta)(T-n)n+f(\alpha)\alpha\sum_{k=1}^{n-1}(n-k)+f(-\beta)(-\beta)\sum_{k=1}^{T-n-1}(T-n-k)\nonumber\\
        & = f(\alpha)(-\beta)(T-n)n+f(\alpha)\alpha[n(n-1)-\frac{n(n-1)}{2}]\nonumber\\&+f(-\beta)(-\beta)[(T-n)(T-n-1)-\frac{(T-n)(T-n-1)}{2}]\nonumber\\
        & = nf(\alpha)(T-n)(-\beta)+f(\alpha)\alpha\frac{n(n-1)}{2}+f(-\beta)(-\beta)\frac{(T-n)(T-n-1)}{2}\\
        & = -\frac{n\alpha}{2}[(n+1)f(\alpha)+f(-\beta)(T-n-1)]
    \end{align*}
    For the temporary impact term
\begin{align*}
    \kappa\sum_{t=1}^{T} v_t^2 & = \kappa[n\alpha^2+(T-n)\beta^2]
\end{align*}
thus
\begin{equation*}
    \E[X_T] = -\bigg[-\frac{n\alpha}{2}[(n+1)f(\alpha)+f(-\beta)(T-n-1)]+\kappa[n\alpha^2+(T-n)\beta^2]\bigg].
\end{equation*}
\end{proof}
We then prove lemma \ref{lem:second lemma}
\begin{proof}
    From Lemma \ref{lem: expected cash}, for any $n\in[T-1]$ the expected cash flow is written as
    \begin{align}\label{eq: E[X_T] with alpha}
        \E[X_T] & = \frac{n\alpha}{2}\Big[(n+1)\theta\alpha^{\delta}-(T-n-1)\Big(\frac{n\alpha}{T-n}\Big)^{\delta}\Big]-\kappa\Big[n\alpha^2+\frac{n^2\alpha^2}{T-n}\Big]\nonumber\\
        & = \alpha^{\delta+1}\underbrace{\frac{n\theta}{2}\Big[(n+1)-(T-n-1)\Big(\frac{n}{T-n}\Big)^{\delta}\Big]}_{A(n)}-\alpha^2\underbrace{\kappa\Big[n+\frac{n^2}{T-n}\Big]}_{B(n)}
    \end{align}
    We observe that $\E[X_T]>0$ if and only if 
    \begin{equation*}
        A(n)>0\Leftrightarrow n+1>(T-n-1)\Big(\frac{n}{T-n}\Big)^{\delta}
    \end{equation*}
     and in that case $\E[X_T]$ is a concave function of $\alpha$, thus to find the optimal $\alpha$ we compute the derivative with respect to $\alpha$ in Eq.(\ref{eq: E[X_T] with alpha}):
     \begin{align*}
         \frac{\partial\E[X_T]}{\partial\alpha} & = (\delta+1)\alpha^{\delta}A(n)-2\alpha B(n)\Leftrightarrow\\
         \alpha^* & = \bigg(\frac{(\delta+1)A(n)}{2B(n)}\bigg)^{\frac{1}{1-\delta}}.
     \end{align*}
     Hence the cash flow of the optimal strategy is then
     \begin{equation}\label{eq: PnD}
        \max_{n\in[T-1]}\E[X_T](\alpha^*,n).
    \end{equation}
\end{proof}


\section{Properties of the local maxima of the optimization problem} \label{subsec: Similarities among Trading Strategies}
We analyze the trading strategies obtained by SLSQP from $N$ different initializations, focusing on how similar the resulting strategies are in terms of their final cash flow. To compare two strategies $\vec{v}_i,\vec{v}_j\in\mathbb{R}^T$, we use the cosine similarity\footnote{The cosine similarity between two round-trip strategies corresponds to their Pearson correlation (see also \cite{Zhelezniak_etal}).} \citep{Intro2DataMining}
\[
c(i,j)=\frac{\langle \vec{v}_i,\vec{v}_j\rangle}{\|\vec{v}_i\|\|\vec{v}_j\|},
\]
which measures the angle between the two vectors and lies in $[-1,1]$. Since a strategy and its sign-reversed counterpart produce the same final cash position, we use the sign-invariant distance
\[
d(i,j)=1-|c(i,j)|.
\]
This distance treats aligned and anti-aligned strategies as equivalent. We then examine the pairwise distance matrices for all SLSQP terminal strategies. As expected (see Fig.~(\ref{fig:distances})), nearby strategies in the ordering tend to have smaller distances. However, the plots alone do not reveal whether similar strategies necessarily yield similar expected cash. Therefore, for all pairs satisfying $d(i,j)\leq \eta$ with $\eta=0.1$, we compute $|\mathbb{E}_i[X_T]-\mathbb{E}_j[X_T]|$ (see Fig.~(\ref{fig:discrepancy_theta0.1})). The resulting scatter plot shows that strategies can be close under the sign-invariant distance while still producing different expected cash, indicating structural differences despite high similarity. Finally in Fig.~(\ref{fig:213 vs 452}) are illustrated two strategies with small sign-invariant distance but relatively large cash discrepancy. 
 
\begin{figure}[h!]
    \centering
  
    \includegraphics[width=0.48\linewidth]{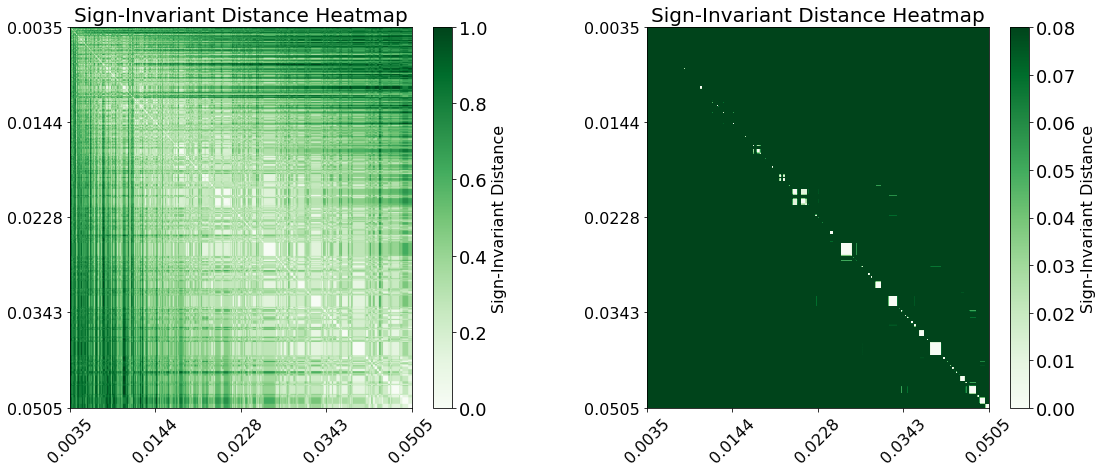}
    \caption{The heatmap of the sign-invariant distance on the left for all the strategies and on the right with precision of two colors for all the strategies in Fig.~(\ref{fig:numSol}).}
    \label{fig:distances}
\end{figure}

\begin{figure}
    \centering
    \includegraphics[width=0.5\linewidth]{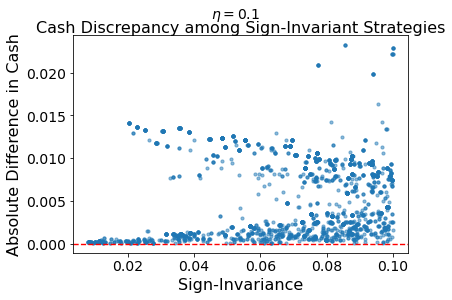}
    \caption{Each point represents a pair of strategies whose sign-invariant distance is less than or equal to the threshold $\eta$.}
    \label{fig:discrepancy_theta0.1}
\end{figure}

\begin{figure}[h!]
    \centering
    \subfigure[]{
    \includegraphics[width=0.45\linewidth]{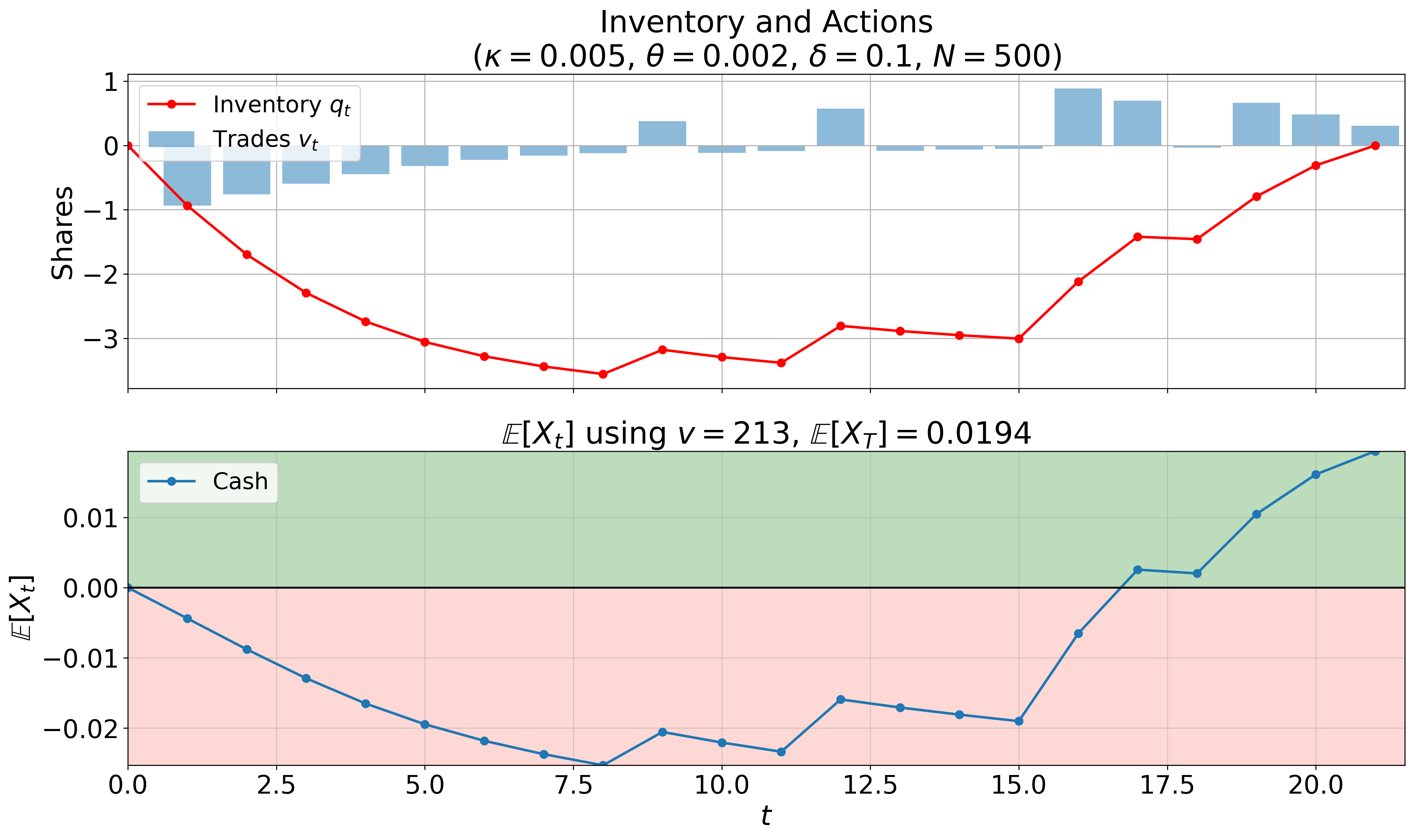}}
    \subfigure[]{
    \includegraphics[width=0.45\linewidth]{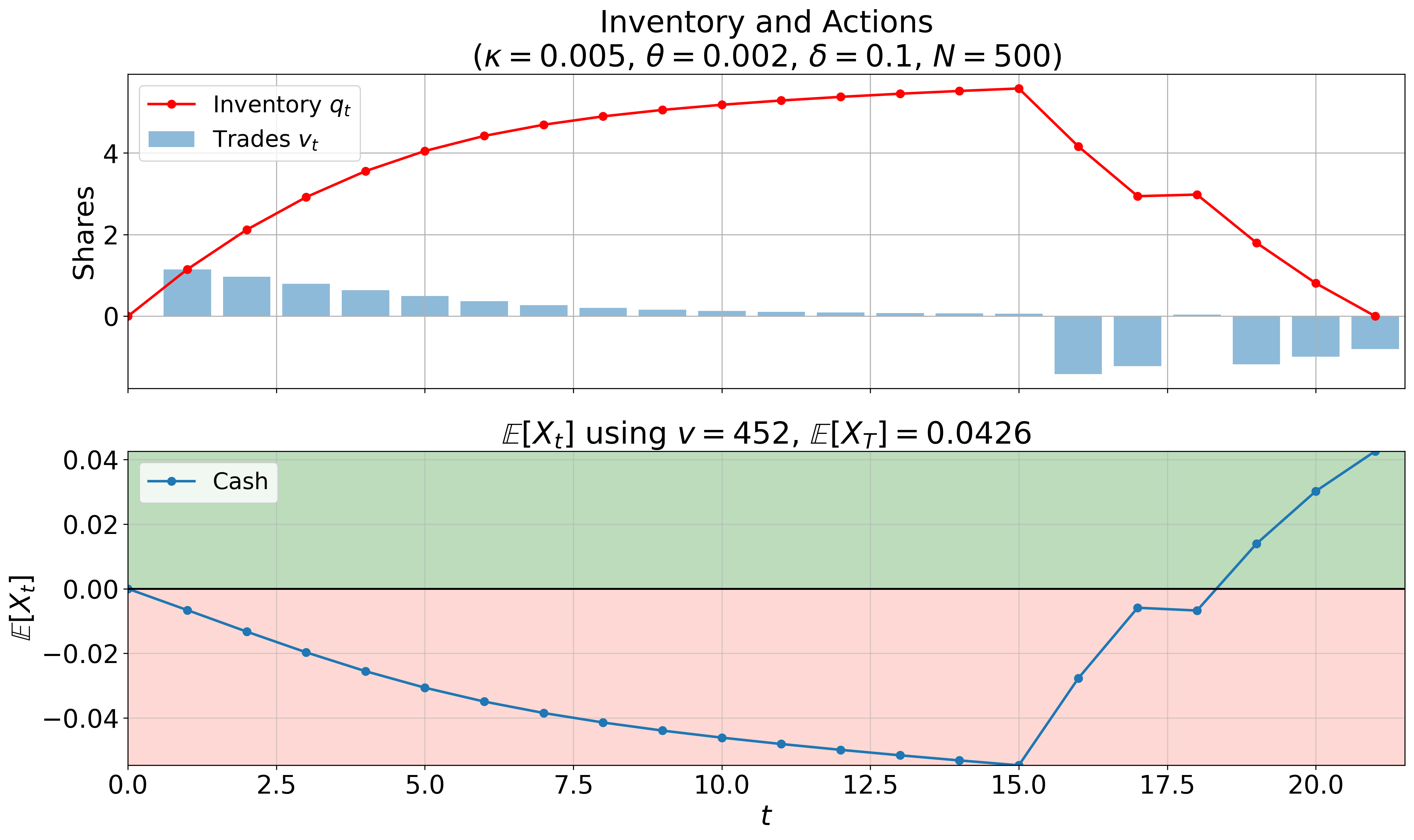}}
    \caption{The two strategies with the largest discrepancy in expected cash for which $d(i,j) = 0.0855$.}
    \label{fig:213 vs 452}
\end{figure}
\section{Estimating the Parameters in SLSQP}\label{app: Estimating the Parameters in SLSQP}
In this Appendix we justify the derivation for the estimation of the permanent impact parameter $\theta$ as a function of the exponent impact parameter $\delta$. Since
\begin{equation}
    \hat{\delta} = \argmin_{\delta\in(0,1)}\underbrace{\sum_{i=1}^N\Big(I_i-\hat{\theta}(\delta)\frac{T}{S_0}\Big(\frac{q_i}{T}\Big)^{\delta}\Big)^2}_{\text{RSS}(\hat{\theta})}
\end{equation}
for the simplification in the notation we use: $\hat{\theta}(\delta):=\hat{\theta}$. The RSS is a concave function in $\hat{\theta}$ thus there exist unique minimizer, it is then enough to consider the first derivative:
\begin{align*}
    \frac{\partial \text{RSS}}{\partial\hat{\theta}} & = 0 \Leftrightarrow\\
    -2\frac{T}{S_0}\sum_{i=1}^N\Big(\frac{q_i}{T}\Big)^{\delta}\Big(I_i-\hat{\theta}\frac{T}{S_0}\Big(\frac{q_i}{T}\Big)^{\delta}\Big) & = 0\Leftrightarrow\\
    \hat{\theta}(\delta) & = \frac{S_0}{T}\frac{\sum_{i=1}^N(\frac{q_i}{T})^{\delta}I_i}{\sum_{i=1}^N(\frac{q_i}{T})^{2\delta}}.
\end{align*}
\begin{figure}
    \centering
    \subfigure[]{\includegraphics[width=1\textwidth]{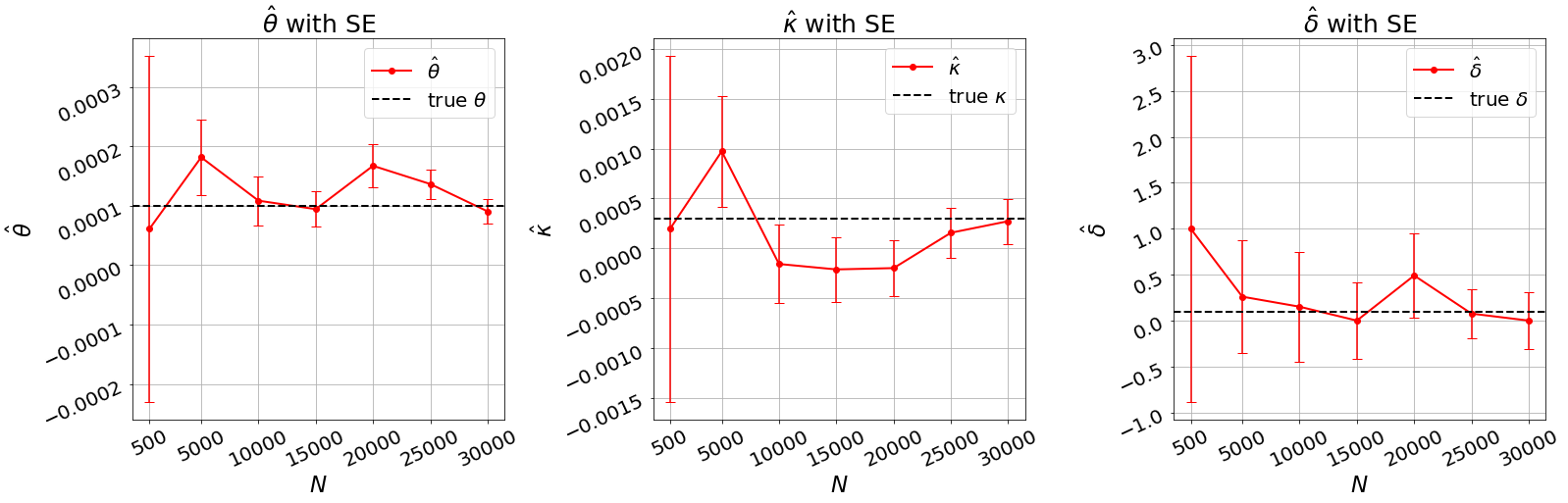}}
    \subfigure[]{\includegraphics[width=1\textwidth]{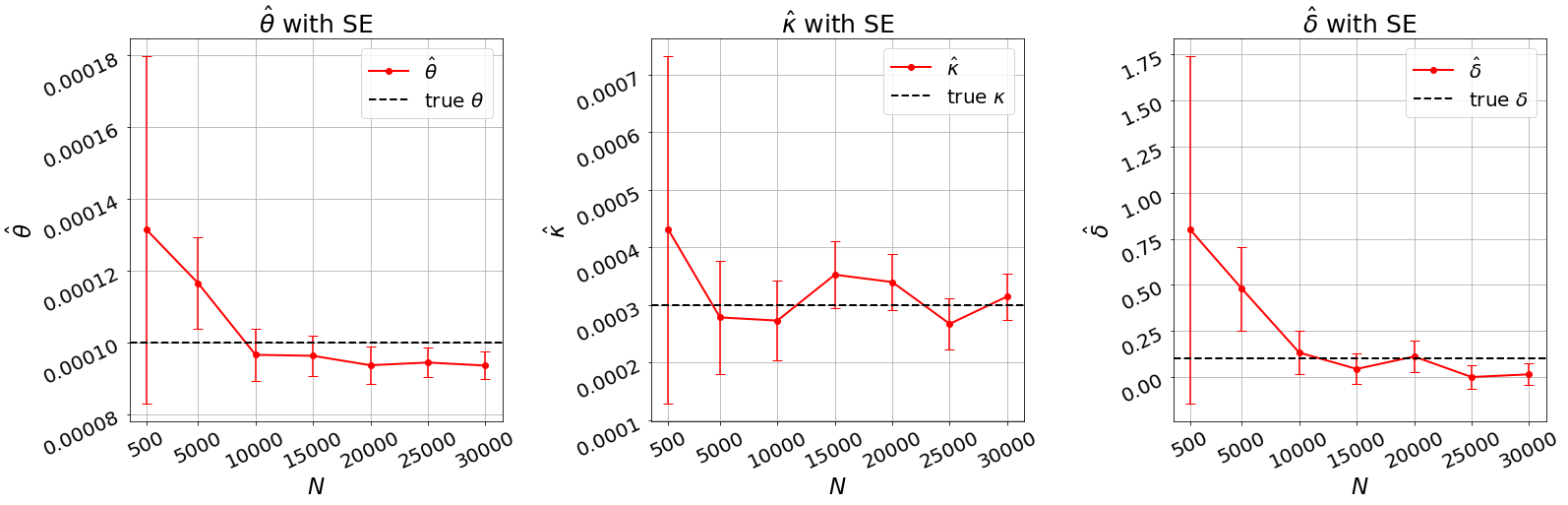}}
    \subfigure[]{\includegraphics[width=1\textwidth]{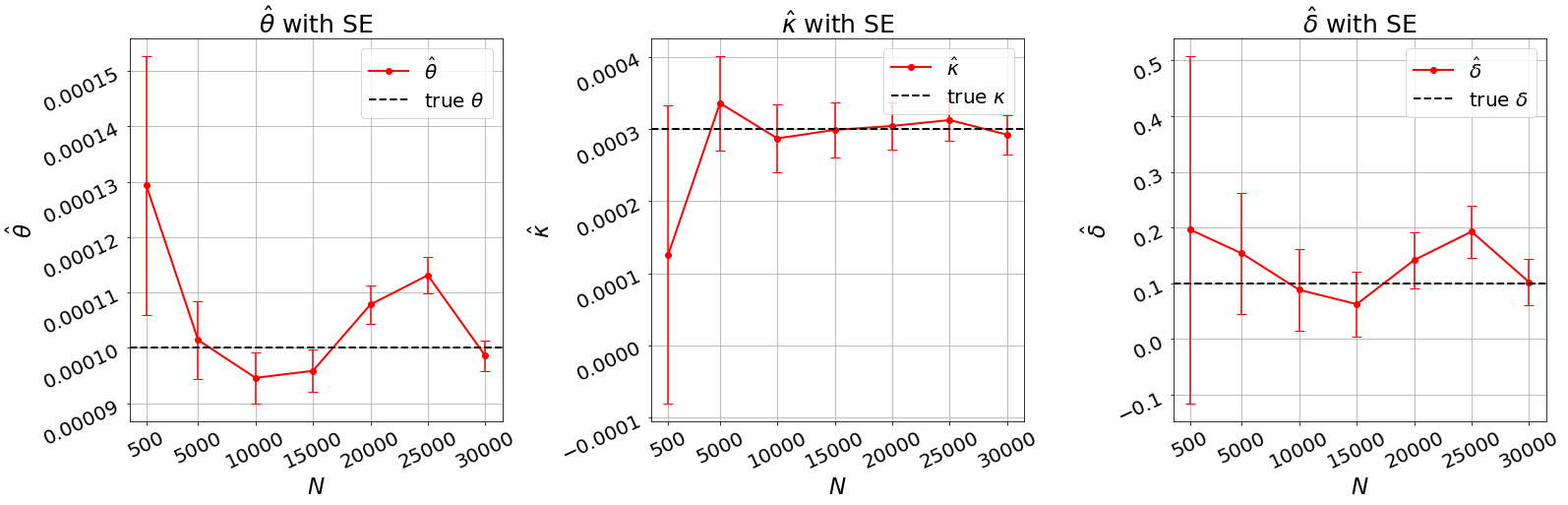}}
    \caption{Estimation of the impact parameters along with their standard errors as a function of the number of metaorders $N$ (a) when $\sigma=0.0168$, (b) when $\sigma=0.003$ and (c) when $\sigma=0.002$.}
    \label{fig:consistent impact params}
\end{figure}
Fig.~(\ref{fig:consistent impact params}) depicts the estimated impact parameters as a function of the number of metaorders $N$ for all three values of $\sigma$. We observe that the estimators are consistent since the standard errors get decrease as $N$ increases and approach the true value of the parameters. Fig.~(\ref{fig:regSLSQP}) illustrates the strategy obtained with the $\text{regSLSQP}_N$ for $N=500$ and $N=5000$ metaorders when the volatility is $0.0168$ in panels (a) and (d), when the volatility is $0.003$ in panels (b) and (e) and when the volatility is $0.002$ in panels (c) and (f). Clearly the expected cash-flow is positive since we use the complete information of the price dynamics for its computation.
\begin{figure}
    \centering
    \subfigure[]{\includegraphics[width=0.32\linewidth]{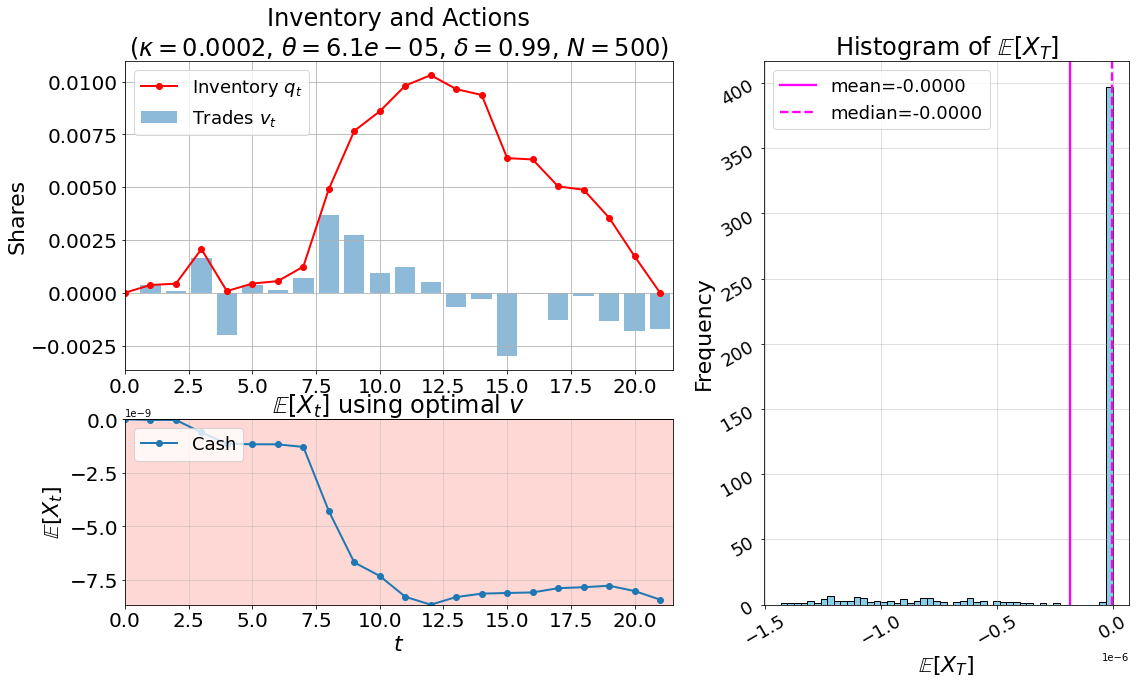}}
    \subfigure[]{\includegraphics[width=0.32\linewidth]{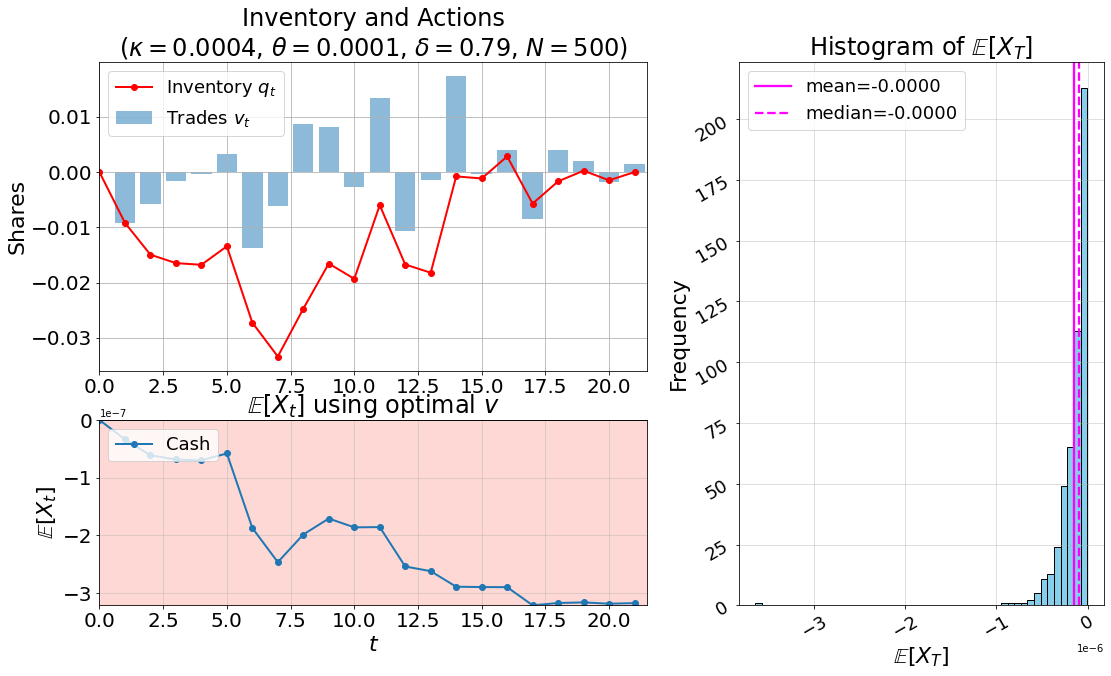}}
    \subfigure[]{\includegraphics[width=0.32\linewidth]{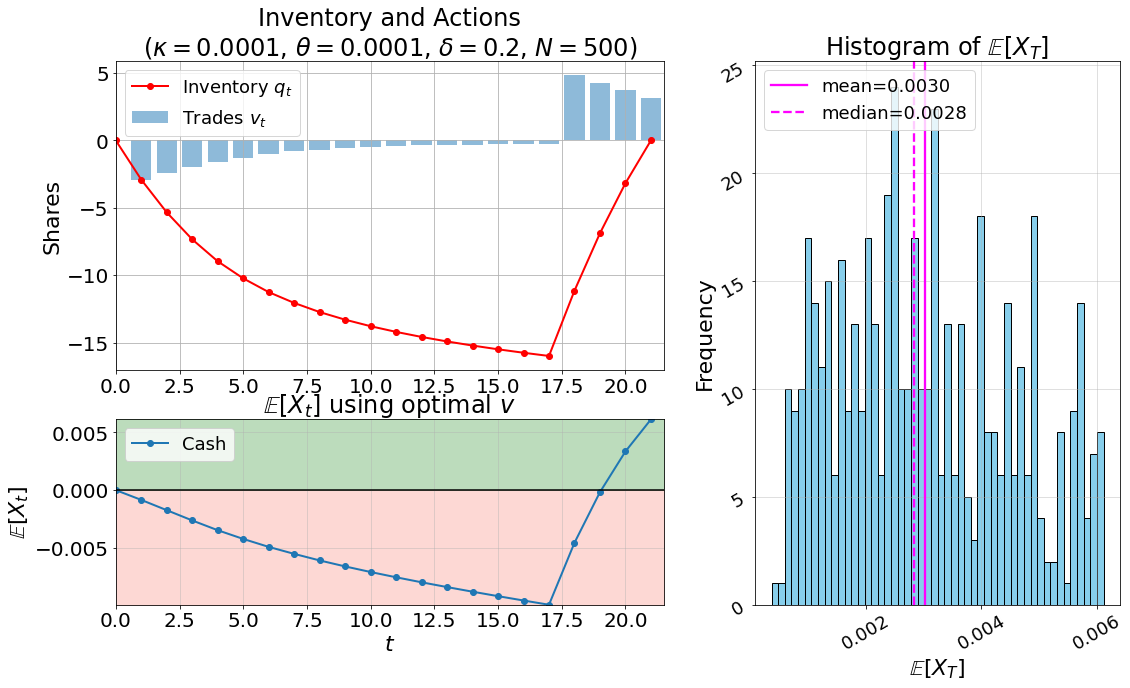}}
    \subfigure[]{\includegraphics[width=0.32\linewidth]{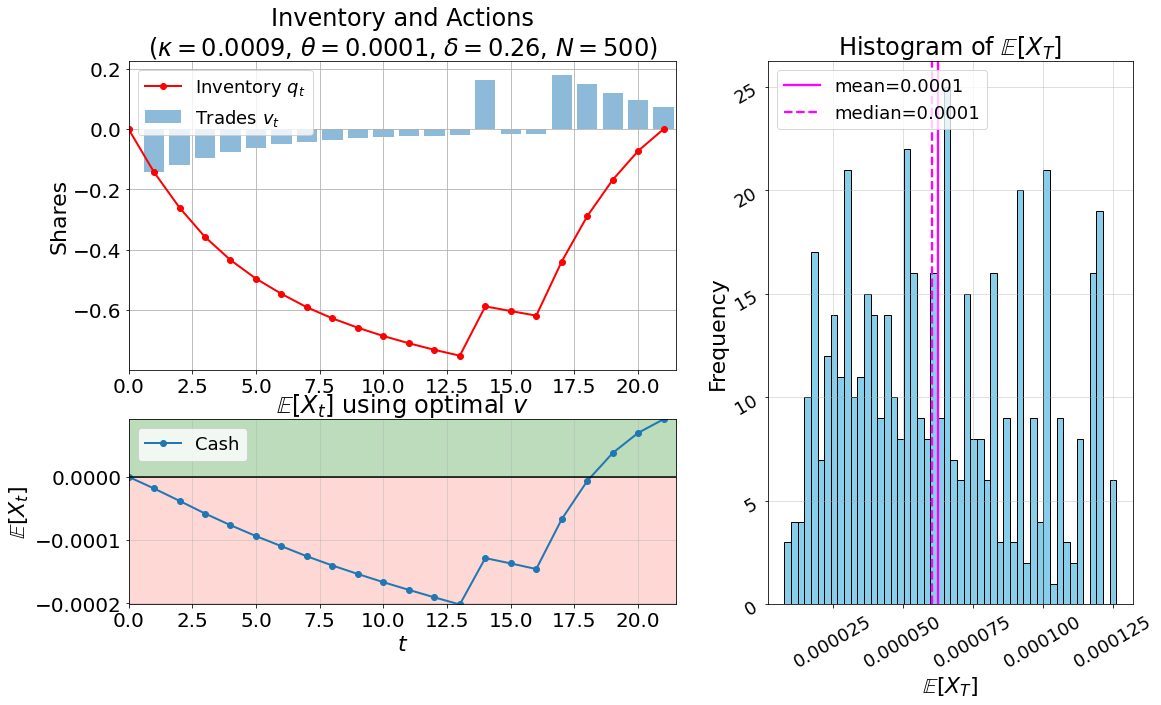}}
    \subfigure[]{\includegraphics[width=0.32\linewidth]{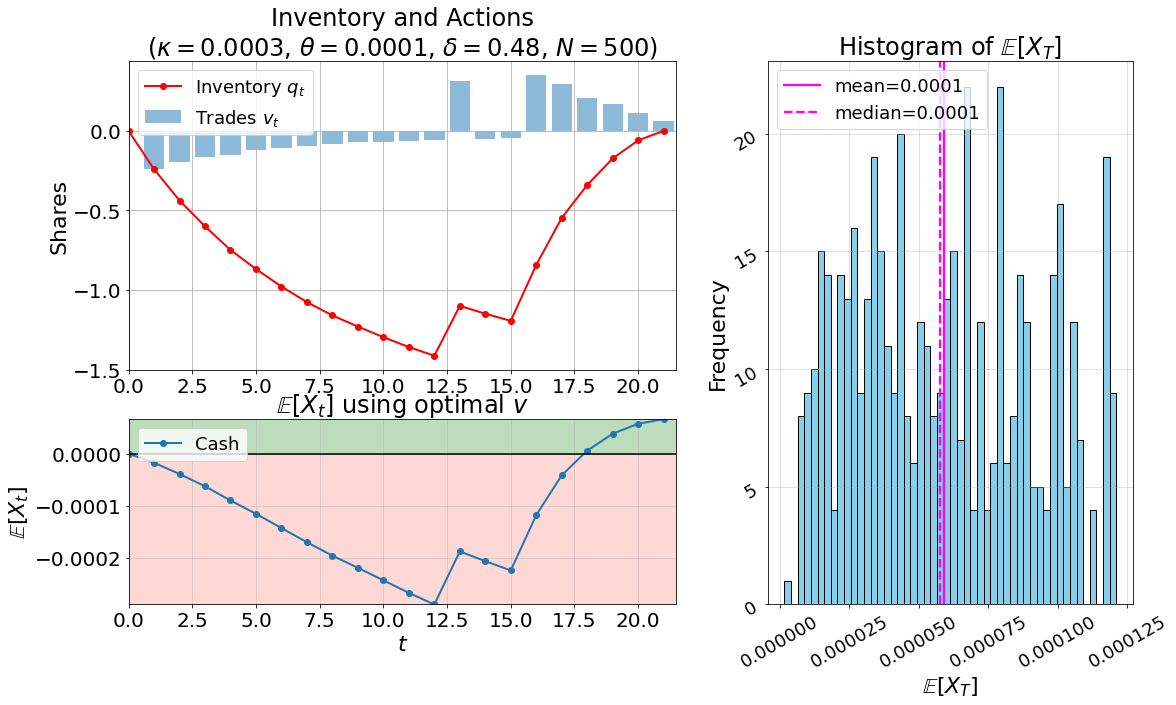}}
    \subfigure[]{\includegraphics[width=0.32\linewidth]{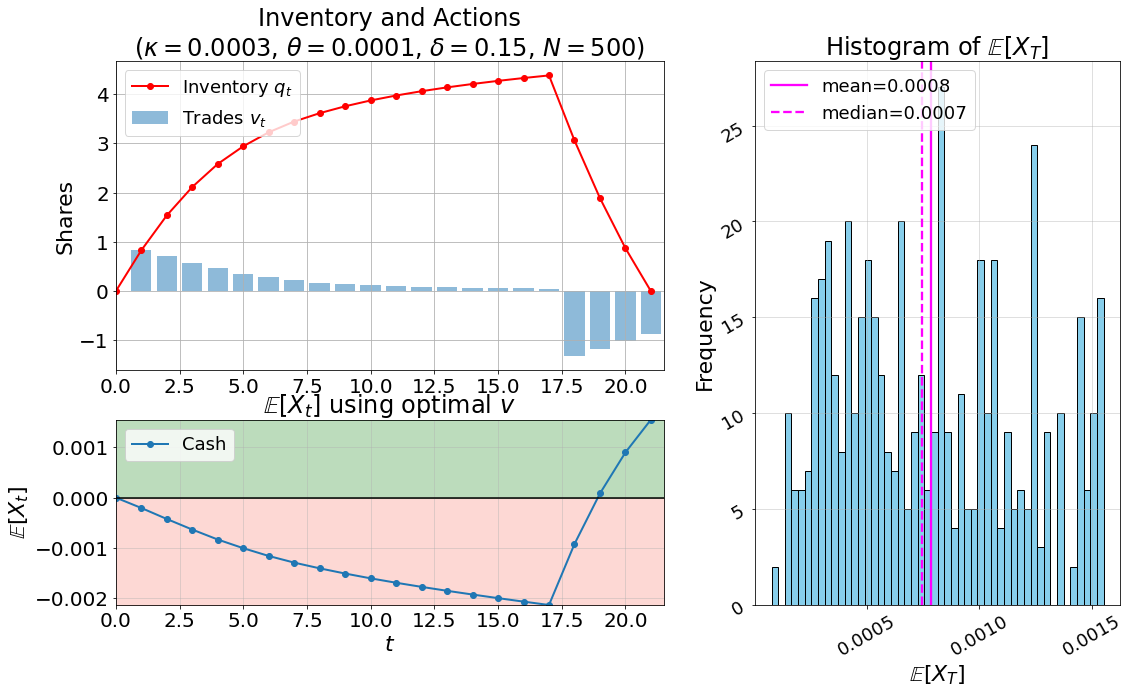}}
    \caption{Plots {\bf (a-c)} when $N=500$ and {\bf (d-f)} when $N=5000$ metaorders. In each panel, {\bf Top-Left:} The optimal trading strategy by $\text{regSLSQP}_N$ and the corresponding inventory (red line). {\bf Bottom-Left:} The expected cash-flow. The green (red) area indicates the region for which $\E[X_t]>0$ ($\E[X_t]<0$). {\bf Right:} The histogram of $\E[X_T]$ from the different starting points.}
    \label{fig:regSLSQP}
\end{figure}
\section{The training phase of DDPG}\label{app: training DDPG}
Here we plot the actions' aggressiveness (see Eq.(\ref{eq: aggressiveness})) on the left plot of Fig.~(\ref{fig:Aggr vs Dirct}) and the actions' directionality (see Eq.(\ref{eq: directionality})) on the right plot of Fig.~(\ref{fig:Aggr vs Dirct}) during the training phase of DDPG when the reward is 
\begin{equation}\label{eq: rev reward}
    r_t = -S_tv_t-\kappa v^2_t
\end{equation}
\begin{figure}
    \centering
    \includegraphics[width=1\linewidth]{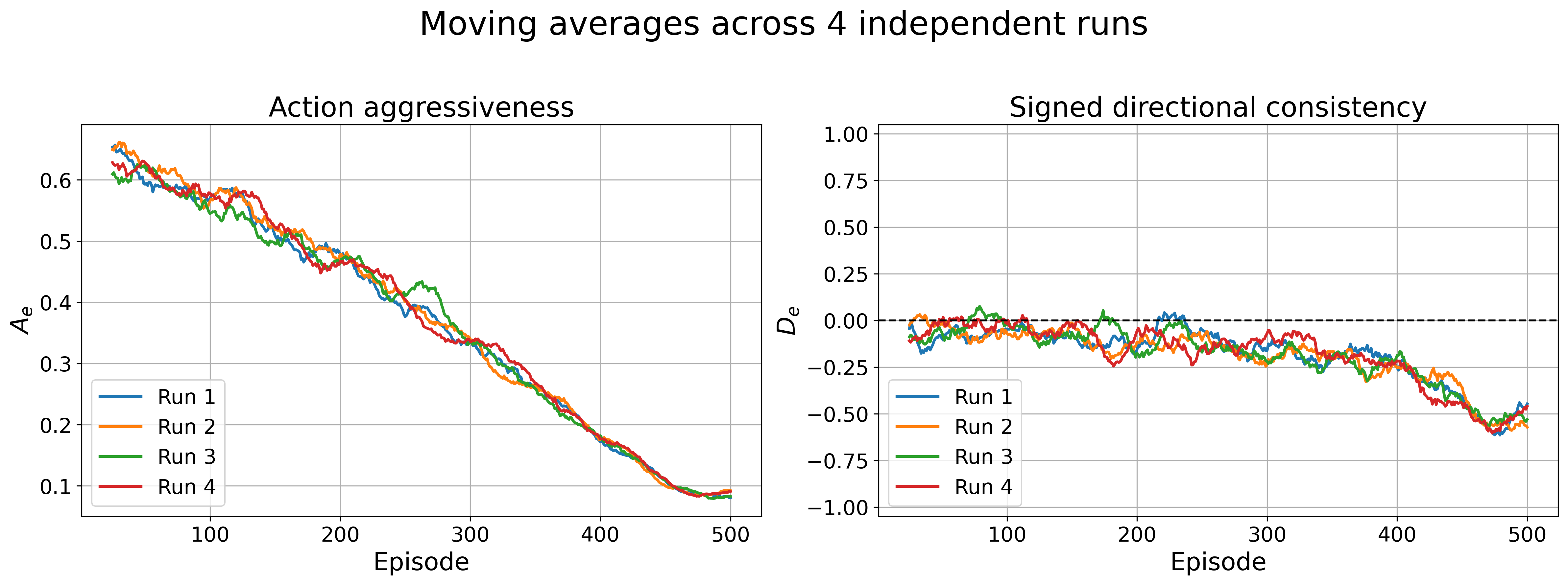}
    \caption{The 25-episode moving averages {\bf Left:} of the action aggressiveness and {\bf Right:} of signed action directionality during training and the reward is defined as in Eq. (\ref{eq: rev reward}).}
    \label{fig:Aggr vs Dirct}
\end{figure}
We observe that, although the level of aggressiveness remains approximately the same throughout the training episodes with that of Fig.~(\ref{fig:Aggresiveness and Directionality}), the directionality of the actions is reversed due to the negative sign in the reward function, relative to the directionality obtained when the reward is defined as in Eq.(\ref{eq:reward}).

\end{document}

%% file: preamble.tex
\usepackage[utf8]{inputenc}
\usepackage{blindtext}
\usepackage{titlesec}
\usepackage{graphicx}
\usepackage{amsmath}
\usepackage{amsthm}
\usepackage{amssymb}
\usepackage{tikz}
\usetikzlibrary{bayesnet}
\usetikzlibrary{arrows.meta}
\usetikzlibrary{backgrounds}
\usepackage{color}
\usepackage{xcolor}
\usepackage{subfigure}
\usepackage{afterpage}
\usepackage{algorithm} 
\usepackage{algpseudocode} 
\usepackage{multirow}
\usepackage{hyperref}
\usepackage{algorithm}
\usepackage{algpseudocode}
\usepackage{xcolor}
\usepackage{wrapfig}
\usepackage{epsfig}
\usepackage{mathtools}
\usepackage{todonotes}
\usepackage{colortbl}

\usepackage{natbib}
\setcitestyle{authoryear,open={(},close={)}}
\usepackage[normalem]{ulem}
\hypersetup{
    colorlinks=true,
    linkcolor= violet,
    citecolor = magenta,
    filecolor= brown,      
    urlcolor= blue,
    pdftitle={Overleaf Example},
    pdfpagemode=FullScreen,
    }
\usepackage[font=small,labelfont=bf]{caption} 
\usepackage{graphicx} 
\usepackage{booktabs}
\usepackage{geometry}
\usepackage{eurosym}
\usepackage{authblk}

\newtheorem{definition}{Definition}[section]
\newtheorem{proposition}{Proposition}[section]

\newtheorem{lemma}{Lemma}[section]

\newcommand{\R}{\mathbb{R}}

\newcommand{\N}{\mathbb{N}}
\newcommand{\E}{\mathbb{E}}

\newcommand{\argmax}{\text{argmax}}
\newcommand{\argmin}{\text{argmin}}

\newcommand{\state}{\mathcal{S}_t}